\documentclass[a4paper,reqno]{amsart}
\usepackage[T1]{fontenc}
\usepackage[utf8]{inputenc}
\usepackage[english]{babel}

\usepackage{amsmath,amsthm}
\usepackage{xfrac}
\usepackage{amssymb}
\usepackage{enumitem}
\usepackage{booktabs} 

\usepackage[dvipsnames]{xcolor}
\usepackage[colorlinks,linkcolor=BrickRed,urlcolor=Blue,citecolor=Blue,pdfusetitle]{hyperref}

\newcommand{\ZZ}{\mathbb{Z}}
\newcommand{\NN}{\mathbb{N}}
\newcommand{\FF}{\mathbb{F}}

\DeclareMathOperator{\distance}{d}
\DeclareMathOperator{\weight}{w}

\newcommand{\norm}[2][]{\operatorname{N}^{#1}_{#2}}

\newcommand{\rdiv}{\mid_r}
\newcommand{\ldiv}{\mid_\ell}
\newcommand{\pesc}[2]{\langle #1,#2\rangle}
\newcommand{\units}[1]{#1^\times}
\DeclareMathOperator{\tforall}{for\ all}

\DeclareMathOperator{\Gal}{Gal}
\DeclareMathOperator{\Inn}{Inn}
\newcommand{\gcrd}[1]{\left( #1 \right)_r}
\newcommand{\lclm}[1]{\left[ #1 \right]_\ell}
\newcommand{\lcrm}[1]{\left[ #1 \right]_r}

\newtheorem{theorem}{Theorem}
\newtheorem{lemma}{Lemma}
\newtheorem{prop}{Proposition}
\newtheorem{corollary}{Corollary}
\theoremstyle{definition}
\newtheorem{defn}{Definition}
\newtheorem{example}{Example}
\theoremstyle{remark}
\newtheorem{remark}{Remark}

\setlist[enumerate]{label=\upshape(\arabic*)}

\begin{document}

\thanks{Research funded under grants PID2023-149565NB-I00 and PRE2020-093254, both financed by the Spanish Research Agency (MCIN/AEI / 10.13039/501100011033), the second one being also financed by the European Social Fund ``FSE invierte en tu futuro''.}

\title[LCPs of skew constacyclic codes]{Linear complementary pairs of skew constacyclic codes}

\author{F. J. Lobillo}
\address{IMAG, CITIC and Dept. of Algebra, University of Granada, Spain}
\email{jlobillo@ugr.es}

\author{José Manuel Muñoz}
\address{Dept. of Algebra, University of Granada, Spain}
\email{munoz@ugr.es}

\begin{abstract}
	Linear complementary pairs (LCPs) of codes have been studied since they were introduced in the context of discussing mitigation measures against possible hardware attacks to integrated circuits. In this situation, the security parameters for LCPs of codes are defined as the (Hamming) distance and the dual distance of the codes in the pair. We study the properties of LCPs of skew constacyclic codes, since their algebraic structure provides tools for studying their duals and their distances. As a result, we give a characterization for those pairs, as well as multiple results that lead to constructing pairs with designed security parameters. We extend skew BCH codes to a constacyclic context and show that an LCP of codes can be immediately constructed from a skew BCH constacyclic code. Additionally, we describe a Hamming weight-preserving automorphism group in the set of skew constacyclic codes, which can be used for constructing LCPs of codes.
\end{abstract}



\maketitle

\section{Introduction}

Given a field \( F \) and two \( F \)-linear codes \( \mathcal C, \mathcal D \subset F^n \), the pair \( (\mathcal C, \mathcal D) \) is said to be a \emph{linear complementary pair (LCP) of codes} if \( F^n \) is the direct sum of the vector spaces \( \mathcal C \) and \( \mathcal D \); that is, if \( \mathcal C \) and \( \mathcal D \) are supplementary.
The concept of linear complementary pairs of codes was introduced in \cite{Ngo15} in order to generalize linear complementary dual codes, defined in \cite{Massey92}, into a wider family of codes which kept relevant properties that allow them to be used for encoding integrated circuits in such a way that hardware Trojan horses against such circuits are forced to be as easy as possible to detect in order to have any chance to be successful.
In the context of this application, two security parameters are relevant for any linear complementary pair of codes \( (\mathcal C, \mathcal D) \): \( d_{\text{Trigger}} \) is the minimum Hamming distance between distinct elements in the dual code of \( \mathcal D \), \( \distance(\mathcal D^\perp) \) (also referred to as the \emph{dual distance of \( \mathcal D \)}), and \( d_{\text{Payload}} \) is the minimum Hamming distance between distinct elements in \( \mathcal C \), \( \distance(\mathcal C) \). They represent the minimum number of symbols that a hardware Trojan horse must be able to access in order to, respectively, extract information about the state of the circuit or modify its state; see \cite[Section II]{Ngo15} for details.

For the purpose of constructing linear complementary pairs of codes with suitable dimensions and security parameters, the proposal in \cite[Section III]{Ngo15} is as follows.
If the maximization of \( d_{\text{Payload}} \) is considered more critical than the one of \( d_{\text{Trigger}} \), then \( \mathcal C \) is chosen among linear codes with known, high enough lower bounds on their minimum Hamming distance. Then, random supplementary codes \( \mathcal D \) are computed until \( d_{\text{Trigger}} \), which has to be computed as the dual distance of \( \mathcal D \), is acceptably high.
If, otherwise, the maximization of \( d_{\text{Trigger}} \) is preferable over the one of \( d_{\text{Payload}} \), then \( \mathcal D^\perp \) is chosen and a similar approach is followed for computing \( \mathcal C \).
Computing a random supplement of a vector space is computationally trivial for the practical purposes considered in \cite{Ngo15}. However, the computation of the minimum Hamming distance of a code is costly and has to be performed once for every attempt. Moreover, the number of expected attempts required in order to randomly find a code whose minimum distance is at least the required security parameter could be considerable, especially if the security parameter is close to the best known lower bounds for codes of the required characteristics.
As a result, asymmetric security parameters are taken in the examples for the proposal in \cite{Ngo15}.

This contrasts with the suggested approach in \cite{Guneri18}, which restricts the codes in the pair to families with known algebraic properties. While this results in less available codes as well as less possible supplements for a given code, the additional structure leads to better knowledge about some relevant properties of the codes, such as the dual of a code and lower bounds on their minimum Hamming distance -- which, depending on the considered family, may also be significantly higher than what a random search through the set of supplements of a code could be expected to yield within a reasonable running time. As a result of such properties, under certain circumstances, starting from a single algebraic code \( \mathcal C \) or \( \mathcal D \), a linear complementary pair \( (\mathcal C, \mathcal D) \) can be constructed in such a way that the remaining security parameter is at least as high as the other one. This means that the task of finding a pair with high enough security parameters is reduced to the easier one of finding a single code with high enough distance or dual distance, which can be done, for example, through known databases of codes of the considered family, or known lower bounds for the distance of codes. Furthermore, additional algebraic structure on linear codes may also imply the existence of known fast nearest-neighbor error-correcting decoding algorithms, which could also be relevant for the application in \cite{Ngo15}.

Following this algebraic approach, this work will consider skew constacyclic codes. These are an immediate generalization of skew cyclic codes, which were introduced in \cite{BGU07} as a non-commutative version of the well-known cyclic codes. As shown in recent work on skew constacyclic codes, such as \cite{BBB21} and \cite{GKLN22}, for each field \( F \), each field automorphism \( \sigma : F \to F \), each length \( n \) and each constant \( \lambda \) in the fixed field of \( \sigma \), the set of skew \( \lambda \)-constacyclic codes is identified with the set of left ideals in a ring \( \mathcal R \) determined by \( F \), \( \sigma \), \( n \) and \( \lambda \).
Many properties of skew constacyclic codes, as well as constructions that result from them, are analogous versions of the ones for non-skew constacyclic codes. For example, see the error-correcting decoding algorithms in \cite{BGU07, GLN17s, GLN17pgz} or the lower bounds for the distance of these codes in \cite{BGU07, GLNN18, ALN21}, which are the skew cyclic analogues of similar constructions for cyclic codes.
On the other hand, the noncommutativity of \( \mathcal R \) leads to the potential existence of multiple supplementary left ideals for each left ideal; that is, each code might have multiple supplementary codes. This contrasts with the situation for non-skew constacyclic codes, described in \cite[Section II]{Guneri18}.

A few criteria for determining whether a pair of skew \( \lambda \)-constacyclic codes constitutes a linear complementary pair of codes are given in Theorem \ref{thm:LCP}. This result generalizes \cite[Theorem 2]{BBB21}, which deals with the case where the pairs are constructed by a code and its dual (for two different notions of dual), which requires \( \lambda^2 = 1 \). We also describe skew BCH constacyclic codes of designed distance in Definition \ref{defnBCHlike}, giving a result that allows the construction of LCPs of skew BCH constacyclic codes with designed security parameters in Theorem \ref{thm:BCHLCP}. By studying Hamming weight-preserving isomorphisms and automorphisms and using Theorem \ref{thm:LCP}, this work also describes a few procedures for constructing LCPs of codes while controlling the security parameters.

\section{Preliminaries}

We shall first define two central concepts for this work. We consider linear codes of length \( n \) over a field \( F \); that is, our codes are vector subspaces in \( F^n \).

\begin{defn}[{\cite[Section II.B]{Ngo15}}]\label{defn:LCP}
	Given two \( F \)-linear codes \( \mathcal C, \mathcal D \subset F^n \), the pair \( (\mathcal C, \mathcal D) \) is said to be a \emph{linear complementary pair (LCP) of codes} if their ambient vector space \( F^n \) is the direct sum of \( \mathcal C \) and \( \mathcal D \).
\end{defn}

We will consider a single security parameter following \cite{Guneri18}. We will define the security parameter for any pair of linear codes, not requiring the pair to be LCP as per the previous definition. The distance function is the usual Hamming metric in \( F^n \): the Hamming distance between two vectors in \( F^n \) is the number of nonzero entries in their difference.

\begin{defn}\label{defn:secparam}
	The \emph{security parameter} of a pair of linear codes \( (\mathcal C, \mathcal D) \) is the minimum of \( \distance(\mathcal C) \) and \( \distance(\mathcal D^\perp) \).
\end{defn}

Thus, if the security parameter as in Definition \ref{defn:secparam} is high enough, so will be the two security parameters, \( d_{\text{Trigger}} \) and \( d_{\text{Payload}} \), considered in \cite[Section II.B]{Ngo15}.

We shall now recall some common tools for studying skew constacyclic codes, mostly following \cite{GLNN18} for their construction. We consider an automorphism \(\sigma\) of finite order \(|\sigma| = \mu\) in the Galois group of the field \( F \). Let \(K = F^\sigma\) be its fixed field. The \(i\)-th truncated norm with respect to \( \sigma \) is defined as
\[
\norm[\sigma]{0}(a) = 1, \quad \norm[\sigma]{i}(a) = a \sigma(a) \cdots \sigma^{i-1}(a), i \geq 1
\]
for each \(a \in F\). Observe that the usual norm associated to the field extension \(F/K\) is
\[
\norm{F/K} = \norm[\sigma]{\mu}.
\]
The notations $\sigma$, $\mu$, $F$, $K$, $\norm[\sigma]{i}$ and $\norm{F/K}$ will be used throughout this paper.

Recall that the ring of skew polynomials \(R = F[x;\sigma]\) is the set of polynomials in \(x\) written with coefficients on the left, with the usual sum and a multiplication derived from the rule \(x a = \sigma(a) x\) for each \(a \in F\). This is an instance of the non-commutative polynomials introduced in the seminal paper \cite{Ore33}. It is well-known, see \cite{Jacobson96}, that $R$ is a left and right Euclidean Domain. In particular, every left and every right ideal of $R$ is principal.
The monic generators of the left ideals $Rf_1 \cap \cdots \cap Rf_k$ and \(Rf_1 + \cdots + Rf_k\) are, respectively, the \emph{least common left multiple} and the \emph{greatest common right divisor} of $f_1, \dots, f_k \in R$, denoted by $\lclm{f_1, \dots, f_k}$ and \(\gcrd{f_1, \dots, f_k}\), i.e.,
\begin{equation}\label{lclmgcrd}
Rf_1 \cap \cdots \cap Rf_k = R \lclm{f_1, \dots, f_k} \text{ and }  Rf_1 + \cdots + Rf_k = R \gcrd{f_1, \dots, f_k}.
\end{equation}
These generators can be computed by means of the extended left Euclidean Algorithm (see \cite[Ch. I, Theorem 4.33]{Bueso/alt:2003}). We use the notation $g \rdiv h$ to declare that $g$ is a \emph{right divisor of} $h$, that is, $Rh \subseteq Rg$. Moreover, \(\gamma \in F\) is said to be a \emph{right root} of \(g \in R\) if \(x-\gamma \rdiv g\). For each \(\gamma \in F\) and every \(g = \sum_i g_i x^i \in R\), it follows from \cite[Lemma 2.4]{LL88} that
\begin{equation}\label{eq:rightev}
g = q(x) (x - \gamma) + \sum_i g_i \norm[\sigma]{i}(\gamma),
\end{equation}
so \(\gamma\) is a right root of \(g = \sum_i g_i x^i\) if and only if \(\sum_i g_i \norm[\sigma]{i}(\gamma) = 0\). Greatest common left divisors, least common right multiples, left divisors and left roots are defined analogously. More specifically,
\begin{equation}\label{eq:leftev}
g = (x - \gamma)q(x) + \sum_i \sigma^{-i}(g_i) \norm[\sigma^{-1}]{i}(\gamma),
\end{equation}
so \(\gamma\) is a left root of \( g \) if and only if \(\sum_i \sigma^{-i}(g_i) \norm[\sigma^{-1}]{i}(\gamma) = 0\).

Multiple concepts from polynomials are readily extended to skew polynomials, such as the degree $\deg(f)$ of a skew polynomial $f$ and the concept of a monomial, which is any element of the form $ax^k$ for some $a \in F$ and $k \in \NN$. As \( R \) is a domain, the degree satisfies
\begin{equation}\label{degequations}
\deg(fg) = \deg(f) + \deg(g) = \deg(\gcrd{f,g}) + \deg(\lclm{f,g})
\end{equation}
for any \( f, g \in R \), as it does for commutative polynomials over a field; see \cite{Ore33} for a primary source, or \cite{Bueso/alt:2003} for a more modern approach.

By \cite[Theorem 1.1.22]{Jacobson96}, the center of \(R\) is \(Z(R) = K[x^\mu]\) where \( \mu = |\sigma| \). Hence, for each \(n = \mu s\) for some \( s \ge 1 \) and each \(\lambda \in K\), \(x^n - \lambda\) is an element in \(Z(R)\). It follows that the ring
\begin{equation}\label{introR}
\mathcal{R} = \frac{R}{R(x^n - \lambda)}
\end{equation}
is a \(K\)-algebra. Its elements can be represented by the skew polynomials in \( R \) of degree less than \(n\), so it is isomorphic to \(F^n\) as an \(F\)-vector space. Its set of units is denoted by \(\units{\mathcal{R}}\).

\begin{defn}[{\cite[Definition 2.2]{GLNN18}}]\label{degreesextension}
We say that \(\sigma\) has an extension \(\theta\) of degree \(s\) if there exists a field extension \(F \subseteq L\) and \(\theta \in \Gal(L)\) such that \(|\theta| = n = s\mu\), \(\theta_{|F}= \sigma\) and \(L^\theta = F^\sigma = K\).
\end{defn}

For such an extension and for each \(a \in F \subseteq L\),
\(
\theta^\mu(a) = \sigma^\mu(a) = a,
\)
so \(F \subseteq L^{\theta^\mu}\). Since \([L:F] = s\) and \(|\theta^\mu| = s\), the equality \(F = L^{\theta^\mu}\) follows. Denoting \(S = L[x;\theta]\) and
\begin{equation}\label{introS}
\mathcal{S} = \frac{S}{S(x^n-\lambda)},
\end{equation}
there is a canonical inclusion
\(
\mathcal{R} \subseteq \mathcal{S}
\)
as a consequence of \cite[Lemma 2.3]{GLNN18}. If \(F\) is a finite field, degree \(s\) extensions always exist, whilst if \(F\) is the field of rational functions over a finite field, they exist when \(\mu\) and \(s\) are coprime. See \cite[Examples 2.4 and 2.5]{GLNN18}.

\section{Skew constacyclic codes}

For convenience, we shall recall that the ring defined in \eqref{introR} is 
\(
\mathcal{R} = \frac{R}{R(x^n-\lambda)}
\)
where \(R = F[x;\sigma]\), \(|\sigma| = \mu\), \(n = s \mu\) and \(\lambda \in K = F^\sigma\). From now on, we will require \( \lambda \) to be nonzero. We shall identify each coset $a \in \mathcal R$ with the unique skew polynomial in \( R \) of degree less than $n$ that is contained in $a$, as the arithmetic in $\mathcal R$ is the one in $R$ modulo $x^n - \lambda$. Additionally, $\mathcal R$ is an $n$-dimensional $F$-vector space, as shown by the canonical $F$-vector space isomorphism
\begin{equation}\label{eq:vsisom}
\begin{aligned}
F^n & \to \mathcal R \\
(a_0, a_1, \dots, a_{n-1}) & \mapsto a_0 + a_1 x + \dots + a_{n-1} x^{n-1}.
\end{aligned}
\end{equation}
\begin{defn}
	A \emph{skew $\lambda$-constacyclic code over $F$ of length $n$} is a left ideal of $\mathcal R$ or, equivalently, a vector subspace in $F^n$ such that its image under the map in \eqref{eq:vsisom} is a left ideal of $\mathcal R$.
	A \emph{skew cyclic code} is a skew $1$-constacyclic code, and a \emph{skew negacyclic code} is a skew $(-1)$-constacyclic code.
\end{defn}

The isomorphism in \eqref{eq:vsisom} induces a Hamming metric in $\mathcal R$: the Hamming weight of $f \in \mathcal R$ is the number of nonzero coefficients of the only skew polynomial in $R$ of degree less than $n$ that is projected into $f$, which is the remainder of the right division of $f$ by $x^n - \lambda$. We will only consider the Hamming metric, as it is the one that suits the application in \cite{Ngo15}. If we know the dimension $k$ and the Hamming distance $d$ of a skew $\lambda$-constacyclic code $\mathcal C$ of length $n$, we say that $\mathcal C$ is an $[n, k, d]$-code. This notation is standard for linear codes, as well as \(\weight(c)\), \(\distance(c,c') = \weight(c-c')\) and \(\distance(\mathcal{C})\) to denote the Hamming weight of \( c \in \mathcal R \), the Hamming distance between \( c, c' \in \mathcal R \) and the minimum distance between distinct elements in \( \mathcal C \), respectively.

Recall that \(R\) is a left and right principal ideal domain. This implies that every left ideal in $\mathcal R$ is principal as well, being also generated by some $g \in R$ (note that \( \mathcal R \) is a right \( R \)-module) such that $g \rdiv x^n - \lambda$, which can be taken as a monic polynomial. The monic generator of each left ideal in $\mathcal R$ is the unique monic polynomial of minimal degree in the left ideal. Therefore, there is a one-to-one correspondence between the set of skew $\lambda$-constacyclic codes in $\mathcal R$ and the set of monic right divisors of $x^n - \lambda$.

Note that, if \( f \) does not right divide \( x^n - \lambda \), then \( \mathcal R f = \mathcal R g \) where \( g = \gcrd{f, x^n - \lambda} \) is monic and right divides \( x^n - \lambda \).
In the sequel, unless otherwise stated, we assume that a skew \(\lambda\)-constacyclic code is given by its monic right divisor of \(x^n-\lambda\). For example, if we state that the code \( \mathcal C \) is \( \mathcal R g \) or is generated by \( g \), we mean that \( g \) is a monic skew polynomial right dividing \( x^n - \lambda \).
The $F$-dimension of the code generated by some skew polynomial $g$ of degree $\deg(g)$ is equal to $n - \deg(g)$.

If \( x^n - \lambda \) admits a decomposition of the form \( x^n - \lambda = \lclm{x - \beta_1, \dots, x - \beta_n} \), then some right divisors of \( x^n - \lambda \) can be constructed as the least common left multiple of subsets of \( \{ x - \beta_1, \dots, x - \beta_n \} \).
We do not need those $\beta_i$ to be in $F$. In fact, if $n = s \mu$ for some $s > 1$, in order to fully understand such possible decompositions of $x^n - \lambda$, we will require the existence of an extension \( \theta : L \to L \) of \(\sigma\) of degree \(s\) according to Definition \ref{degreesextension}.

Using the notation of \cite{LL88,LL04,DL07}, a skew polynomial which can be decomposed as least common left multiple of linear skew polynomials is called a Wedderburn polynomial (W-polynomial for short). Moreover, \(\{\beta_1, \dots, \beta_k\} \subseteq L\) are called P-independent if \(\deg(\lclm{x-\beta_1, \dots, x - \beta_k}) = k\) (see \cite[Definition 2.3]{DL07}), so a skew polynomial \( f \in S = L[x;\theta] \) of degree \( k \) is a W-polynomial when \( f = \lclm{x-\beta_1, \dots, x - \beta_k} \) for some P-independent set \(\{\beta_1, \dots, \beta_k\} \subseteq L\) (see \cite[Theorem 5.8, i) iff iii)]{LLO07}). As \( S \) is a domain and \( (x^n - \lambda) \) is in its center, the proof of the following lemma is straightforward.

\begin{lemma}\label{lemma:ghhg}
	For any \(g, h \in S \), \( gh = x^n - \lambda \) if and only if \( hg = x^n - \lambda \).
\end{lemma}

The following result shows that either all factors of \( x^n - \lambda \) including itself can be decomposed as the least common left multiple of linear skew polynomials, where \( g \) is a factor of \( x^n - \lambda \) if \( x^n - \lambda = f_1 g f_2 \) for some \( f_1, f_2 \in S \), or no factors of positive degree can be decomposed in such a way.

\begin{prop}\label{prop:decomposition}
	The following statements are equivalent:
	\begin{enumerate}
		\item \label{prop:decomposition-factor} Every monic factor \( g \in S \) of the skew polynomial \( x^n - \lambda \) is a W-polynomial.
		\item \label{prop:decomposition-lambda} \( x^n - \lambda \) is a W-polynomial.
		\item \label{prop:decomposition-anyfactor} There exists some monic factor \( g \) of \( x^n - \lambda \) of degree at least \( 1 \) which is a W-polynomial.
		\item \label{prop:decomposition-u} There exists $u \in L$ such that $\norm{L/K}(u) = \lambda$, or equivalently, \( x - u \rdiv x^n - \lambda \).
	\end{enumerate}
	If some monic factor \( g \) of degree \( k \) is decomposed as \( g = \lclm{x - \beta_1, \dots, x - \beta_k} \), then there exist \( \alpha_1, \dots, \alpha_k \in L \) linearly independent over \( K \) such that \( \beta_i = \alpha_i^{-1}u\theta(\alpha_i) \) for each \( 1 \le i \le k \); in particular, \( \norm{L/K}(\beta_i) = \lambda \).
\end{prop}

\begin{proof}
	It is immediate that \ref{prop:decomposition-factor} implies \ref{prop:decomposition-lambda}. By \cite[Theorem 5.1]{LL04}, if \(x^n-\lambda\) is a W-polynomial then every monic factor of \(x^n-\lambda\) is a W-polynomial. This means that \ref{prop:decomposition-lambda} implies \ref{prop:decomposition-factor}.

	By choosing \( g = x^n - \lambda \), \ref{prop:decomposition-lambda} implies \ref{prop:decomposition-anyfactor}. Lemma \ref{lemma:ghhg} implies that any factor of \( x^n - \lambda \) is a right divisor of \( x^n - \lambda \). As a result, any \( \beta_i \) in the decomposition of any factor \( g \) of \( x^n - \lambda \) is a right root of \( x^n - \lambda \), hence \( \norm{L/K}(\beta_i) = \lambda \). We may take any right root \( \beta_i \) of \( x^n - \lambda \) or any factor thereof as \( u \). Thus, \ref{prop:decomposition-anyfactor} implies \ref{prop:decomposition-u}.

	In order to prove that \ref{prop:decomposition-u} implies \ref{prop:decomposition-lambda}, and as a result the four statements are equivalent, take any \( \alpha_1, \dots, \alpha_n \in L \) conforming a \( K \)-basis of \( L \). By \cite[Theorem 5.3, (i) iff (iii)]{DL07}, if we define \( \beta_i = \alpha_i^{-1} u \theta(\alpha_i) \), then the set \( \{ \beta_1, \dots, \beta_n \} \) is P-independent and has \( n \) right roots of \( x^n - \lambda \), therefore \( \lclm{x - \beta_1, \dots, x - \beta_n} = x^n - \lambda \) as both are monic of degree \( n \) and the former right divides the latter. Hence, \ref{prop:decomposition-u} implies \ref{prop:decomposition-lambda}.

	Given a factor \( g \) of degree \( k \) and the P-independent set from its decomposition \( \beta_1, \dots, \beta_k \), by virtue of Hilbert's 90 Theorem (see \cite[Chapter VI, Theorem 6.1]{Lang02}), as \(\norm{L/K}(\beta_i / u) = \norm{L/K}(\beta_i) \norm{L/K}(u)^{-1} = 1\) for \(1 \leq i \leq k\), there exist some \(\alpha_1, \dots, \alpha_k \in L\) such that \(\beta_i = \alpha_i^{-1} u \theta(\alpha_i)\). Their \( K \)-linear independence is a result of \cite[Theorem 5.3]{DL07} and the P-independence of \(\beta_1, \dots, \beta_k\).
\end{proof}

\begin{remark}\label{remark:existenceoflambda}
	For any extension of finite fields $L/K$, the norm function $\norm{L/K} : L \to K$ is surjective as shown in \cite[Theorem 2.28]{Lidl97}, i.e., for any $\lambda \in K$, there is some $u \in L$ such that $\norm{L/K}(u) = \lambda$. This is not the case for a general field $K$. If no such element \( u \in L \) exists, then it is immediate that there are no such decompositions as the ones given in Proposition \ref{prop:decomposition}.

	In \cite{GKLN22}, skew \( \lambda \)-constacyclic codes where \( \lambda \) is not a norm for a field extension \( F/F^\sigma \) where \( F \) is of the form \( \FF_q(t) \) are constructed.
The ideas in \cite{GKLN22}, especially the one of embedding \( S \) into a skew polynomial ring over the splitting field of \( x^n - \lambda \), could be relevant for extending the sections in this work that require \( \lambda \) to be a norm into these contexts.
\end{remark}

Any norm map sends the $1$ of a field into itself. Thus, in the skew cyclic case, where $\lambda = 1$, we may take $u = 1$ and decompositions as in Proposition \ref{prop:decomposition} always exist. This is exploited in \cite{GLN16}, where a skew cyclic setting is considered, which, by carefully choosing such decompositions, paved the way for constructing skew cyclic codes with known lower bounds for their Hamming, as well as rank, distance. For example, see \cite{GLN17s}, \cite{GLN17pgz}, \cite{GLNN18} or \cite{ALN21}. Note that these works take advantage of the existence of right roots of the generator for the code.

The existence of some \( u \in L \) of norm \( \lambda \) also allows the following definition, which is taken from \cite[Lemma 3.2]{DL07}.

\newcommand{\E}[2]{E({#1},{#2})}
\begin{defn}
	The \( K \)-vector space \( \E g u \subseteq L \), for any monic \( g \in S \) and any \( u \in L \) such that \( g \rdiv x^n - \lambda \) and \( \norm{L/K}(u) = \lambda \), is defined as
	\[
		\E g u = \{ 0 \} \cup \{ \alpha \in \units L : x - \alpha^{-1} u \theta(\alpha) \rdiv g \}.
	\]
\end{defn}

The set \( \E g u \) is a vector space as it is shown in \cite[Lemma 3.2]{DL07} to be the kernel of the $K$-linear map $g(T_u)$, where $T_u : L \to L$ is defined as $a \mapsto \theta(a) u$ and, if $g = \sum_{i=0}^k a_i x^i$, then $g(T_u) = \sum_{i=0}^k a_i T_u^i$, where $T_u^0$ is the identity map in $L$.

\begin{lemma}\label{lemma:kgdim}
	The \( K \)-dimension of $\E g u$ is equal to $\deg(g)$.
\end{lemma}

\begin{proof}
	Let \( d = \deg(g) \). By Proposition \ref{prop:decomposition}, \( g = \lclm{x - \alpha_1^{-1} u \theta(\alpha_1), \dots, x - \alpha_d^{-1} u \theta(\alpha_d)} \) for some \( K \)-linearly independent \( \{\alpha_1,\dots,\alpha_d \} \), which are elements in \( \E g u \). Thus, the $K$-dimension of $\E g u$ is at least $d$. It is also at most $d$ by \cite[Corollary 5.4]{DL07}.
\end{proof}

\begin{lemma}\label{lemma:kggcd}
	For any two monic $g, h \in S$ such that $g, h \rdiv x^n - \lambda$, $\E {\gcrd{g,h}} u = \E g u \cap \E h u$ and $\E {\lclm{g, h}} u = \E g u + \E h u$.
\end{lemma}

\begin{proof}
	Any nonzero element in \( L \) is a right root of both $g$ and $h$ if and only if it is a right root of $\gcrd{g,h}$. Therefore, $\E {\gcrd{g,h}} u = \E g u \cap \E h u$.

	Any right root of either $g$ or $h$ is a right root of $\lclm{g, h}$. Thus, $\E g u \cup \E h u \subseteq \E {\lclm{g, h}} u$, so $\E g u + \E h u$ is a subspace of $\E {\lclm{g, h}} u$.
	By Lemma \ref{lemma:kgdim} and \eqref{degequations},
	\begin{align*}
		\dim(\E g u + \E h u)
		& = \dim(\E g u) + \dim(\E h u) - \dim(\E g u \cap \E h u) \\
		& = \deg(g) + \deg(h) - \deg(\gcrd{g,h}) \\
		& = \deg(\lclm{g, h}) \\
		& = \dim(\E {\lclm{g, h}} u).
	\end{align*}
	Hence, \(\E g u + \E h u = \E {\lclm{g,h}} u\).
\end{proof}

\begin{corollary}\label{cor:kgh}
	For any monic $g, h \in S$ such that $g, h \rdiv x^n - \lambda$, $\E g u = \E h u$ if and only if $g = h$.
\end{corollary}

\begin{proof}
	By Lemma \ref{lemma:kggcd}, $\E g u = \E h u$ if and only if $\E g u = \E {\gcrd{g,h}} u$. This is true if $g = h$ and, by Lemma \ref{lemma:kgdim}, it is false if $\deg(\gcrd{g,h}) < \deg(g)$, which is equivalent to $g \ne h$ as both $g$ and $h$ are monic.
\end{proof}

\begin{prop}\label{prop:correspondence}
	Assume $0 \le k \le n$ and the existence of some $u \in L$ such that $\norm{L/K}(u) = \lambda$. For any $K$-vector subspace $V \subseteq L$ of dimension $k$, there is exactly one monic right divisor $g \in S$ of $x^n - \lambda$ such that $\E g u = V$, and it is such that $\deg(g) = k$. Therefore, there is a one-to-one correspondence between the set of monic right divisors of $x^n - \lambda$ in $S$ of degree $k$, the set of skew $\lambda$-constacyclic codes in $\mathcal S$ of $L$-dimension $n-k$, and the set of $K$-vector subspaces in $L$ of $K$-dimension $k$.
\end{prop}

Note that not all skew $\lambda$-constacyclic codes in $\mathcal S$ are guaranteed to be in $\mathcal R$, unless $\mathcal S = \mathcal R$ (that is, $s = 1$), just as not all monic right divisors of $x^n - \lambda$ in $S$ are guaranteed to be in $R$.

\begin{proof}
	Take a $K$-basis $\{\alpha_1,\dots,\alpha_k \}$ of $V$. By \cite[Theorem 5.3]{DL07}, $\{ \beta_1,\dots,\beta_k \}$, where $\beta_i = \alpha_i^{-1}u\theta(\alpha_i)$ for each $i = 1,\dots,k$, is a P-independent set, therefore $g = \lclm{x - \beta_1, \dots, x - \beta_k}$ is a monic skew polynomial of degree $k$ right dividing $x^n - \lambda$. The set $\{\alpha_1,\dots,\alpha_k \}$ is a subset in $\E g u$ which spans $V$, so $V \subseteq \E g u$, and they are equal as they have the same $K$-dimension due to Lemma \ref{lemma:kgdim}. The uniqueness of $g$ follows from Corollary \ref{cor:kgh}.
\end{proof}

We conclude this section with its main result, which will play an essential role in the remaining sections.

\begin{theorem}\label{thm:LCP}
	Let $\mathcal C, \mathcal D \subset \mathcal R$ be two skew $\lambda$-constacyclic codes generated by $g$ and $h$ respectively. The following are equivalent:
	\begin{enumerate}
		\item \label{LCP1} $(\mathcal C, \mathcal D)$ is a linear complementary pair of codes (that is, $\mathcal C \oplus \mathcal D = \mathcal R$).
		\item \label{LCP2} $\mathcal C + \mathcal D = \mathcal R$ and $\dim(\mathcal C) + \dim(\mathcal D) = n$.
		\item \label{LCP3} $\mathcal C \cap \mathcal D = \{0\}$ and $\dim(\mathcal C) + \dim(\mathcal D) = n$.

		\item \label{LCP4} $\gcrd{g,h} = 1$ and $\lclm{g, h} = x^n - \lambda$.
		\item \label{LCP5} $\gcrd{g,h} = 1$ and $\deg(g) + \deg(h) = n$.
		\item \label{LCP6} $\lclm{g, h} = x^n - \lambda$ and $\deg(g) + \deg(h) = n$.
	\end{enumerate}

	If \( \sigma \) has some extension \( \theta : L \to L \) as in Definition \ref{degreesextension} and there is some $u \in L$ such that $\norm{L/K}(u) = \lambda$, then the following statements are also equivalent to the previous ones:
	\begin{enumerate}\setcounter{enumi}{6}
		\item \label{LCP7} $\E g u \oplus \E h u = L$.
		\item \label{LCP8} $\E g u \cap \E h u = \{0\}$ and $\deg(g) + \deg(h) = n$.
		\item \label{LCP9} $\E g u + \E h u = L$ and $\deg(g) + \deg(h) = n$.
	\end{enumerate}

	The conditions $\deg(g) + \deg(h) = n$ and $\dim(\mathcal C) + \dim(\mathcal D) = n$ are equivalent and might replace each other.
\end{theorem}

\begin{proof}
The equivalences between \ref{LCP1}, \ref{LCP2} and \ref{LCP3} are immediate.
The equivalences between \ref{LCP4}, \ref{LCP5} and \ref{LCP6} follow from Equation \eqref{degequations}.
The equivalence between \ref{LCP1} and \ref{LCP4} follows from Equation \eqref{lclmgcrd}.
Lemmas \ref{lemma:kgdim} and \ref{lemma:kggcd} yield the remaining equivalences.
\end{proof}

If \( s = 1 \) (equivalently, \( n = \mu \), \( L = F \), \( S = R \)) and \( \norm{F/K}(u) = \lambda \) for some \( u \in F \), then Theorem \ref{thm:LCP} gives a simple way to construct a linear complementary pair of skew constacyclic codes from a single skew \( \lambda \)-constacyclic code \( \mathcal R g \), as it amounts to finding a basis for a supplementary space for \( \E g u \), which can be done by completing a \( K \)-basis for \( \E g u \) into a \( K \)-basis for \( F \). For \( s > 1 \), even if \( \norm{L/K}(u) = \lambda \) for some \( u \in L \), this approach is not guaranteed to yield a generator in \( R \). In any case, such an approach would disregard the value for the security parameter, which we are interested in, as the distance for the constructed supplement can be expected to follow the distribution of distances of random skew \( \lambda \)-constacyclic codes of the required dimension.

\section{Duals of skew constacyclic codes}

The security parameter for a linear complementary pair of codes, as described in Definition \ref{defn:secparam}, involves the dual of the second code in the pair. The notion of duality for skew constacyclic codes comes from the structure of \( F \)-vector space given by \eqref{eq:vsisom}.
When considering skew constacyclic codes, we will take advantage of the fact that the dual of a skew \( \lambda \)-constacyclic code is a skew \( \lambda^{-1} \)-constacyclic code, as shown by the next result, which summarizes some of the findings in \cite[Section 3]{GLN19t}.
We define the ring
	\begin{equation}\label{eq:hatR}
		\widehat{\mathcal R} = \frac R{R(x^n - \lambda^{-1})},
	\end{equation}
whose left ideals are skew \( \lambda^{-1} \)-constacyclic codes of length \( n \), as well as the anti-isomorphism of rings \( \Theta : \mathcal R \to \widehat {\mathcal R} \) defined over monomials as
	\[
		\sum_{i=0}^{n-1} a_i x^i \mapsto \sum_{i=0}^{n-1} \sigma^{-i}(a_i) x^{-i},
	\]
	where $x^{-1}$ is the reciprocal of the unit $x$ in $\widehat{\mathcal R}$, that is, $x^{-1} = \lambda x^{n-1}$.

\begin{defn}\label{defn:monicrec}
	For any \( h = \sum_{i=0}^{k} a_i x^i \in R \) where \( k < n \) and \( a_k, a_0 \) are nonzero, we define \( h^\Theta \) as the monic skew polynomial \( \sigma^k(a_0)^{-1}x^k \Theta(h) \).
\end{defn}

Note that, if $h = \sum_{i=0}^{k} a_i x^i$ where \( k < n \), we may calculate $x^k \Theta(h)$ as \begin{equation} \label{eq:dualgenerator}
\begin{aligned}
	x^k \Theta\left(\sum_{i=0}^{k} a_i x^i\right) & = \sum_{i=0}^k x^k \sigma^{-i}(a_i) x^{-i} \\
	& = \sum_{i=0}^k \sigma^{k-i}(a_i) x^{k-i} \\
	& = \sum_{i=0}^k \sigma^i(a_{k-i}) x^i.
\end{aligned} \end{equation}

As a result, \( h^\Theta \) is a monic skew polynomial of the same degree as \( h \), and \( h^\Theta \) can be defined over skew polynomials in \( R \) of degree less than \( n \). Note that the constant coefficient of a skew polynomial right (or left) dividing \( x^n - \lambda \) cannot be zero.

\begin{prop}\label{prop:dual}
	If $\mathcal C = \mathcal Rg$ for some monic
	$g \in R$ such that $hg = x^n - \lambda$ for some $h \in R$, then $\mathcal C^\perp = \widehat {\mathcal R} \Theta(h) = \widehat {\mathcal R} x^k \Theta(h) = \widehat {\mathcal R} h^\Theta$, where
	\( k = \deg(h) = n - \deg(g)\).
\end{prop}

\begin{proof}
	See \cite[Section 3]{GLN19t}, especially Theorem 29 and Example 33.
\end{proof}

\begin{remark}\label{remark:lambdapm1}
	If $\lambda^2 = 1$, then $\widehat{\mathcal R} = \mathcal R$. Hence, the dual of a skew cyclic code is a skew cyclic code, and the dual of a skew negacyclic code is a skew negacyclic code. Furthermore, some skew cyclic and skew negacyclic codes are supplementary to their own dual, resulting in linear complementary dual (LCD) codes. LCD skew constacyclic codes over finite fields are studied in \cite{BBB21}.
\end{remark}

If there is an extension \( \theta : L \to L \) of \( \sigma \) of degree \( s \) as described in Definition \ref{degreesextension}, it results in an extension of \( \Theta \) into the anti-isomorphism of rings
\begin{equation}\label{eq:ThetaS}
	\Theta : \mathcal S \to \widehat{\mathcal S} = \frac{S}{S(x^n - \lambda^{-1})},
	\qquad
	\sum_{i=0}^{n-1} a_i x^i \mapsto \sum_{i=0}^{n-1} \theta^{-i}(a_i) x^{-i}.
\end{equation}
If the statements in Proposition \ref{prop:decomposition} are true, then we can get some additional information on the dual of a skew constacyclic code. The key to do so is the following lemma, which is a generalization of \cite[Lemma 35]{GLN19t}.

\begin{lemma}\label{lemma:35ish}
	Let \( \gamma \in L \) be such that \( x - \gamma \) left divides \( x^n - \lambda \). Then \( \gamma' = \theta(\gamma)^{-1} \) is such that \( x - \gamma' \) right divides \( x^n - \lambda^{-1} \) and \( \widehat{\mathcal S}(x - \gamma') = \widehat{\mathcal S}\Theta(x - \gamma) \).
	Moreover, if \( h = \lcrm{x - \gamma_i ~|~ 1 \le i \le k} \) where \( \gamma_1, \dots, \gamma_k \) are left roots of \( x^n - \lambda \), then \( \widehat{\mathcal S}\Theta(h) = \widehat{\mathcal S} \hat h \) where \[ \hat h = \lclm{x - \theta(\gamma_i^{-1}) ~|~ 1 \le i \le k}. \]
\end{lemma}

\begin{proof}
If \(\gamma\) is a left root of \(x^n-\lambda\), then it is also a right root by Lemma \ref{lemma:ghhg}, so \(\lambda = \norm{L/K}(\gamma)\).
The multiplicativity of both \(\theta\) and \(\norm{L/K}\) and the fact that \(\theta(\lambda) = \lambda\) imply that \(\lambda^{-1} = \norm{L/K}(\theta(\gamma)^{-1})\), or equivalently, \(\theta(\gamma)^{-1} = \gamma'\) is a right root of \(x^n - \lambda^{-1}\). The right ideal \((x - \gamma) \mathcal{S}\) is maximal, so its image under \(\Theta\), \(\widehat{\mathcal{S}} \Theta(x - \gamma)\), is a maximal left ideal. Since
\[
-\theta(\gamma)^{-1} x \Theta(x - \gamma) = -\theta(\gamma)^{-1} x (x^{-1} - \gamma) = x - \gamma',
\]
we conclude that \(\widehat{\mathcal{S}}(x - \gamma') \subseteq \widehat{\mathcal{S}} \Theta(x - \gamma)\). The maximality of \(\widehat{\mathcal{S}}(x - \gamma')\) implies the equality.
Finally, assume that \( h = \lcrm{x - \gamma_i ~|~ 1 \le i \le k} \). Then,
\[
\begin{split}
\widehat{\mathcal{S}} \Theta(h) &= \Theta \left( h \mathcal{S} \right) \\
&= \Theta \left( \bigcap_{1 \le i \le k} (x - \gamma_i) \mathcal{S}) \right) \\
&= \bigcap_{1 \le i \le k} \Theta( (x - \gamma_i) \mathcal{S}) \\
&= \bigcap_{1 \le i \le k} \widehat{\mathcal S}\Theta(x - \gamma_i) \\
&= \bigcap_{1 \le i \le k} \widehat{\mathcal S}(x - \theta(\gamma_i^{-1})) \\
&= \widehat{\mathcal{S}} \lclm{\{x - \theta(\gamma_i^{-1}) ~|~ 1 \le i \le k\}}
\end{split}
\]
as asserted.
\end{proof}

\begin{prop}\label{prop:dual2}
	Let \( h, g \in S \) such that \( h = \lcrm{x - \gamma_i ~|~ 1 \le i \le k} \) and \( hg = x^n - \lambda \).
	Then, \( (\mathcal S g)^\perp = \widehat{\mathcal S} \hat h \) where \( \widehat{\mathcal S} \) is as in \eqref{eq:ThetaS} and
	\( \hat h = \lclm{x - \theta(\gamma_i^{-1}) ~|~ 1 \le i \le k}. \)
	In addition, if \( g \in R \), then \( \hat h \in R \) and \( (\mathcal R g)^\perp = \widehat{\mathcal R} \hat h \) where \( \widehat{\mathcal R} \) is as in \eqref{eq:hatR}.
\end{prop}
\begin{proof}
	By Proposition \ref{prop:dual}, \( (\mathcal S g)^\perp = \widehat{\mathcal S} \Theta(h) \), which by Lemma \ref{lemma:35ish} equals \(\widehat{\mathcal S} \hat h \).

	Now assume that \( g \in R \). As \( hg = x^n - \lambda \in R \), this implies that \( h \in R \) and \( \Theta(h) \in R \). Proposition \ref{prop:dual} also proves that \( (\mathcal R g)^\perp = \widehat{\mathcal R} \Theta(h) \). The monic, minimum degree generators as left ideals of \( \widehat{\mathcal R} \Theta(h) \) and \( \widehat{\mathcal S} \Theta(h) \) can both be computed as \( \gcrd{\Theta(h), x^n - \lambda} \), which is an element in \( R \) and is known to be \( \hat h \) in the second case. Hence, \( (\mathcal R g)^\perp = \widehat{\mathcal R} \hat h \).
\end{proof}

\begin{remark}
	By \cite[Theorem 4.5]{LamLeroy2000}, a skew polynomial is a W-polynomial if and only if it is decomposable as a least common \emph{right} multiple of linear skew polynomials. As a result, whether a decomposition of any left divisor \( h \) of \( x^n - \lambda \) exists, as required in this result, is also given by the equivalent statements in Proposition \ref{prop:decomposition}.
\end{remark}

\section{Constructing LCPs with BCH-like codes}\label{LCP-BCH}

In this section we provide a method for constructing skew constacyclic BCH-like codes with designed distance, as well as linear complementary pairs of such codes with designed security parameter, as per Definition \ref{defn:secparam}. We keep the same framework from previous sections.

\begin{lemma}\label{rightrootsR}
Let \(g \in R\) and let \(\gamma \in L\). If \(\gamma\) is a right root of \(g\), then \(\theta^\mu(\gamma)\) is also a right root of \(g\).
\end{lemma}

\begin{proof}
If \(g = \sum_i g_i x^i\), \(\gamma\) is a right root of \(g\) if and only if \(\sum_i g_i \norm[\theta]{i}(\gamma) = 0\) by \eqref{eq:rightev}. Since \(g_i \in F = L^{\theta^{\mu}}\),
\[
0 = \theta^\mu\left( \sum_i g_i \norm[\theta]{i}(\gamma) \right) = \sum_i \theta^{\mu}(g_i) \norm[\theta]{i}(\theta^{\mu}(\gamma)) = \sum_i g_i \norm[\theta]{i}(\theta^{\mu}(\gamma)).
\]
Therefore, \(\theta^\mu(\gamma)\) is also a right root of \(g\).
\end{proof}

As a consequence of Lemma \ref{rightrootsR}, if \(\gamma \in L\) is a right root of \(g\) it follows that
\[
\lclm{x - \theta^{i\mu}(\gamma) ~|~ 0 \leq i \leq s-1} \rdiv g.
\]
These elements are the building blocks of a decomposition of \(g\) in \(R\) as the following lemma shows.

\begin{lemma}\label{lclminR}
For all \(\gamma \in L\), \(\lclm{x - \theta^{i\mu}(\gamma) ~|~ 0 \leq i \leq s-1} \in R\).
\end{lemma}

\begin{proof}
The automorphism \(\theta : L \to L\) can be extended to a bijection
\begin{equation}\label{eq:thetaL}
\begin{split}
\theta : L[x;\theta] &\to L[x;\theta] \\
\textstyle \sum_i a_i x^i &\mapsto \textstyle \sum_i \theta(a_i) x^i.
\end{split}
\end{equation}
From the equality \(F = L^{\theta^\mu}\), it follows that for each \(f \in S = L[x;\theta]\), \(f \in R = F[x;\sigma]\) if and only if \(\theta^\mu(f) = f\). Since for all \(a \in L\)
\[
\theta(xa) = \theta(\theta(a) x) = \theta^2(a) x = x \theta(a) = \theta(x) \theta(a),
\]
\cite[Proposition 2.4]{Goodearl/Warfield:2004} implies that \eqref{eq:thetaL} is a ring automorphism. This means that \(\lclm{x - \theta^{i\mu}(\gamma) ~|~ 0 \leq i \leq s-1} \in R\) as
\begin{multline*}
\theta^{\mu} \left( \lclm{x - \theta^{i\mu}(\gamma) ~|~ 0 \leq i \leq s-1} \right) = \lclm{x - \theta^{(i+1)\mu}(\gamma) ~|~ 0 \leq i \leq s-1} \\
= \lclm{x - \theta^{i\mu}(\gamma) ~|~ 1 \leq i \leq s} = \lclm{x - \theta^{i\mu}(\gamma) ~|~ 0 \leq i \leq s-1},
\end{multline*}
since the order of \(\theta\) is \(n = s \mu\).
\end{proof}

The elements \(\{\gamma, \theta^{\mu}(\gamma), \dots, \theta^{(s-1)\mu}(\gamma)\}\) are not necessarily P-independent, so the degree of \(\lclm{x - \theta^{i\mu}(\gamma) ~|~ 0 \leq i \leq s-1}\) could be strictly less than \(s\).

We will apply Theorem \ref{thm:LCP} starting with a decomposition of \(x^n - \lambda \in S\). This requires, by Proposition \ref{prop:decomposition}, some \(u \in L\) such that \(\lambda = \norm[\theta]{n}(u)\). Let \(\alpha \in L\) be such that \(\left\{ \theta^i(\alpha) \norm[\theta]{i}(u) ~|~ 0 \leq i \leq n-1 \right\}\) is a \(K\)-basis for \(L\). Such \(\alpha\) is called a cyclic vector since it is a generator of \(L\) as a \(K[x]\)-module, where the action is given by \(x a = \theta(a) u\). A cyclic vector can be easily obtained by a random search, see \cite[Appendix]{GNS21} for further details. Following \cite{LL88}, we denote by $a^b$, for each \(a,b \in L\), \(b \neq 0\), the $\theta$-conjugate of $a$ by $b$:
\begin{equation}\label{eq:conjugate}
	a^b = \theta(b) a b^{-1}.
\end{equation}
Within this section, the conjugate will always be with respect to $\theta$.
It follows from \cite[Theorem 5.3]{DL07} and \cite[Lemma 12]{GNS21} that
\begin{equation}\label{x^n-lambda_decomposition}
x^n - \lambda = \lclm{x - u^{\theta^i(\alpha) \norm[\theta]{i}(u)} ~|~ 0 \leq i \leq n-1},
\end{equation}
i.e., the elements \(\left\{ u^{\theta^i(\alpha) \norm[\theta]{i}(u)} ~|~ 0 \leq i \leq n-1 \right\}\) are P-independent. The next formula, which holds for any nonzero \( u, \alpha \in L \) and any \( i \in \NN \), eases the proofs of the forthcoming results.
\begin{equation}\label{theta^i(u^alpha)}
\begin{split}
u^{\theta^i(\alpha) \norm[\theta]{i}(u)} &= \theta \left( \theta^i(\alpha) \norm[\theta]{i}(u) \right) u \left( \theta^i(\alpha) \norm[\theta]{i}(u) \right)^{-1} \\
&= \theta^{i+1}(\alpha) \norm[\theta]{i}(\theta(u)) u \left( \theta^i(\alpha) \norm[\theta]{i}(u) \right)^{-1} \\
&= \theta^{i+1}(\alpha) \norm[\theta]{i+1}(u) \norm[\theta]{i}(u)^{-1} \theta^{i}(\alpha)^{-1} \\
&= \theta^{i+1}(\alpha) \theta^i(u) \theta^{i}(\alpha)^{-1} \\ &= \theta^i(u^\alpha).
\end{split}
\end{equation}

\begin{prop}\label{propBCHlike}
Let \(2 \leq \delta \leq \mu\), \(0 \leq r \leq n-1\) and
\[
g = \lclm{x - \theta^{i+j\mu}(u^\alpha) ~|~ r \leq i \leq r + \delta - 2, 0 \leq j \leq s-1}.
\]
Then \(g \in R\). Moreover, the skew \(\lambda\)-constacyclic code \(\mathcal{C} = \mathcal{R}g\) has \( F \)-dimension \(n - s\delta + s\) and bounded minimum Hamming distance \(\distance(\mathcal{C}) \geq \delta\).
\end{prop}

\begin{proof}
Let \(g_i = \lclm{x - \theta^{j\mu}(\theta^i(u^\alpha)) ~|~ 0 \leq j \leq s-1}\) for each \(r \leq i \leq r + \delta - 2\).
By Lemma \ref{lclminR}, \(g_i \in R\), so \( g \in R \) since \(g = \lclm{g_i ~|~ r \leq i \leq r + \delta - 2}\).
By \cite[Theorem 5.3]{DL07} and \eqref{theta^i(u^alpha)}, \(\left\{ \theta^i(u^\alpha) ~|~ 0 \leq i \leq n-1 \right\}\) is a P-independent set, so \(\deg(g) = s(\delta-1)\), which implies that \(\dim_F(\mathcal{C}) = n - s \delta + s\).

Let
\[
\overline{g} = \lclm{x - \theta^{i}(u^\alpha) ~|~ r \leq i \leq r + \delta - 2}, \mathcal{D} = \mathcal{S}g, \overline{\mathcal{D}} = \mathcal{S} \overline{g}.
\]
Then \(\mathcal{C} \subseteq \mathcal{D}\) as a subfield subcode, and \(\mathcal{D} \subseteq \overline{\mathcal{D}}\), since \(\overline{g} \rdiv g\). It follows that
\begin{equation}\label{eqBCHlike3}
\distance(\mathcal{C}) \geq \distance(\overline{\mathcal{D}}).
\end{equation}
Observe that \(\left\{ \theta^i(\alpha) \norm[\theta]{i}(u) ~|~ 0 \leq i \leq n-1 \right\}\) is a \(K\)-basis for \(L\) if and only if the set \(\left\{ \theta^r(\theta^i(\alpha) \norm[\theta]{i}(u)) ~|~ 0 \leq i \leq n-1 \right\}\) is also a \(K\)-basis. Hence, up to replacing \(\alpha\) and \(u\) by \(\theta^r(\alpha)\) and \(\theta^r(u)\) we do not lose generality if we assume \(r = 0\). Following the notation in \cite{GNS21}, let \(\varphi_u : L \to L\) be the additive map defined by \(\varphi_u(a) = \theta(a) u\). Straightforward computations using \eqref{theta^i(u^alpha)} show that
\[
\varphi_u^j(\alpha) = \theta^j(\alpha) \norm[\theta]{j}(u),
\]
\[\theta^i \left( u^{\varphi_u^j(\alpha)} \right) = \varphi_u^{i+j}(\alpha)^{-1} \varphi_u^{i+j+1}(\alpha)
\]
and
\[
\norm[\theta]{i}\left( u^{\varphi_u^j(\alpha)} \right) = \varphi_u^j(\alpha)^{-1} \varphi_u^{i+j}(\alpha).
\]
The description of \(\overline{g}\) allows the construction of a parity check matrix for \(\overline{\mathcal{D}}\) using right evaluations. Such a matrix is
\[
\Big[ \norm[\theta]{i} \left( u^{\varphi_u^j(\alpha)} \right) \Big]_{\substack{0 \leq i \leq n-1 \\ 0 \leq j \leq \delta-2}} = \Big[ \varphi_u^j(\alpha)^{-1} \varphi_u^{i+j}(\alpha) \Big]_{\substack{0 \leq i \leq n-1 \\ 0 \leq j \leq \delta-2}}.
\]
Since
\[
\Big[ \varphi_u^j(\alpha)^{-1} \varphi_u^{i+j}(\alpha) \Big]_{\substack{0 \leq i \leq n-1 \\ 0 \leq j \leq \delta-2}} = \Big[ \varphi_u^{i+j}(\alpha) \Big]_{\substack{0 \leq i \leq n-1 \\ 0 \leq j \leq \delta-2}} \begin{bmatrix} \varphi_u^0(\alpha)^{-1} & & \\ & \ddots & \\ & & \varphi_u^{\delta-2}(\alpha)^{-1} \end{bmatrix},
\]
the matrix
\[
\Big[ \varphi_u^j(\alpha)^{-1} \varphi_u^{i+j}(\alpha) \Big]_{\substack{0 \leq i \leq n-1 \\ 0 \leq j \leq \delta-2}} = \Big[ \varphi_u^{i+j}(\alpha) \Big]_{\substack{0 \leq i \leq n-1 \\ 0 \leq j \leq \delta-2}}
\]
is also a parity check matrix for \(\overline{\mathcal{D}}\). Therefore \(\overline{\mathcal{D}} = C_{(\varphi_u,\alpha,\delta)}\) in the notation of \cite[Definition 2]{GNS21}. By \cite[Theorem 13]{GNS21}, \(\distance(\overline{\mathcal{D}}) = \delta\) and the proposition follows from \eqref{eqBCHlike3}.
\end{proof}

\begin{defn}\label{defnBCHlike}
Let \(2 \leq \delta \leq \mu\), \(0 \leq r \leq n-1\),
\[
g = \lclm{x - \theta^{i+j\mu}(u^\alpha) ~|~ r \leq i \leq r + \delta - 2, 0 \leq j \leq s-1},
\]
and \(\mathcal{C} = \mathcal{R}g\). The \([n, n - s\delta + s]\)-code \(\mathcal{C}\) is called a skew BCH \(\lambda\)-constacyclic code of designed distance \(\delta\).
\end{defn}

\begin{remark}\label{remark:skewRS}
	Proposition \ref{propBCHlike} and Definition \ref{defnBCHlike} can also be applied to a context where \( \theta \) has order \( n \), and therefore \( s = 1, \mu = n \). In that case, the resulting code is an \( [n, n - \delta + 1] \) skew \( \lambda \)-constacyclic code of designed Hamming distance \( \delta \) -- an MDS code, as the value \( \delta \) for the lower bound given by Proposition \ref{propBCHlike} is equal to the one resulting from the Singleton upper bound. These MDS codes generalize skew Reed--Solomon codes as described in \cite{GLN17pgz}, which match the case where \( u = 1 \) and therefore \( \lambda = 1 \).
\end{remark}

We shall now extend \cite[Theorem 37]{GLN19t} for skew BCH constacyclic codes by showing that the dual of a skew BCH \( \lambda^{-1} \)-constacyclic code (see Proposition \ref{prop:dual}) whose parameters, including a BCH lower bound on its Hamming distance, can be known. The following result, which generalizes \cite[Lemma 36]{GLN19t}, allows us to do so.

\begin{lemma}\label{lemma:36ish}
	Let \( \gamma \in L \) be such that
	\begin{equation}\label{eq:gamma36ish}
		(x - \gamma)\lclm{x - \theta(u^\alpha), \dots, x - \theta^{n-1}(u^\alpha)} = x^n - \lambda
	\end{equation}
	and let \( T, T^c \subseteq \{ 0, \dots, n-1 \} \) be such that \( \{ 0, \dots, n-1 \} \) is the disjoint union of \( T \) and \( T^c \). Then,
	\begin{equation}\label{eq:TcT}
		\lcrm{x - \theta^i(\gamma) ~|~ i \in T^c} \lclm{x - \theta^i(u^\alpha) ~|~ i \in T} = x^n - \lambda.
	\end{equation}
\end{lemma}

Note that a unique \( \gamma \) satisfying \eqref{eq:gamma36ish} exists, as \( \lclm{x - \theta(u^\alpha), \dots, x - \theta^{n-1}(u^\alpha)} \) has degree \( n - 1 \) and right divides \( \lclm{x - \theta^i(u^\alpha) ~|~ 0 \le i \le n-1} \), which equals \( x^n - \lambda \) as a result of \eqref{x^n-lambda_decomposition} and \eqref{theta^i(u^alpha)}.

\begin{proof}
	We define \( N_i = \lclm{x - \theta^j(u^\alpha) ~|~ 0 \le j \le n-1, j \ne i} \) for every \( 0 \le i \le n - 1 \).
	Equation \eqref{eq:gamma36ish} can be written as \( (x - \gamma)N_0 = x^n - \lambda \) which, by applying \( \theta^i \) (as defined in \eqref{eq:thetaL}), results in the identities \( (x - \theta^i(\gamma)) N_i = x^n - \lambda \) for all \( 0 \le i \le n - 1 \).

	We now define, for any set \( A \subseteq \{ 0, \dots, n-1 \} \), the skew polynomials \( P_A = \lclm{x - \theta^i(u^\alpha) ~|~ i \in A} \) and \( Q_A = \lcrm{x - \theta^i(\gamma) ~|~ i \in A} \in S \). This allows us to write \eqref{eq:TcT} as \( Q_{T^c} P_T = x^n - \lambda \).

	\( P_T \) equals \( x^n - \lambda \) if and only if \( T^c \) is empty, in which case \eqref{eq:TcT} holds. Otherwise, there exists some monic, nonconstant \( h \in S \) such that \( hP_T = x^n - \lambda \) as, for any \( k \in T^c \), \( P_T \rdiv N_k \rdiv x^n - \lambda \). Thus, for each \( k \in T^c \) there exists some \( p_k \in S \) such that \( N_k = p_k P_T \), hence \( (x - \theta^k(\gamma)) p_k P_T = x^n - \lambda = h P_T \), so  \( (x - \theta^k(\gamma)) p_k = h \). As a result, \( (x - \theta^k(\gamma)) \ldiv h \), therefore \( Q_{T^c} \ldiv h \). As, in addition, both \( h \) and \( Q_{T^c} \) are monic and the degree of \( h \) is \( n - \deg(P_T) = n - |T| = |T^c| \) where \( |A| \) denotes the cardinality of the set \( A \), if we show that \( \deg(Q_{T^c}) = |T^c| \) we will conclude that \( h = Q_{T^c} \) and therefore \eqref{eq:TcT} holds. We will do so by induction on the number of elements in \( T^c \).

	If \( T^c \) is empty, then \( Q_{T^c} = 1 \) and its degree is \( 0 \).
	Assume that \( T^c = \bar T^c \cup \{ k \} \) and \( \bar T^c \) is such that \( \deg(Q_{\bar T^c}) = |\bar T^c| \), which means that \( Q_{\bar T^c} P_{\bar T} = x^n - \lambda \) for a matching \( \bar T \).
	Either \( \deg(Q_{T^c}) = 1 + \deg(Q_{\bar T^c}) \) or \( \deg(Q_{T^c}) = \deg(Q_{\bar T^c}) \). If the latter is true, then \( Q_{\bar T^c} = Q_{T^c} \) as they are both monic of the same degree and \( Q_{\bar T^c} \ldiv Q_{T^c} \). This means that \( Q_{\bar T^c} = (x - \theta^{k}(\gamma)) p \) for some \( p \in S \), so \( (x - \theta^{k}(\gamma)) p P_{\bar T} = Q_{\bar T^c} P_{\bar T} = x^n - \lambda  = (x - \theta^{k}(\gamma)) N_k \).
	Hence, \( p P_{\bar T} = N_k \), so \( x - \theta^i(u^\alpha) \rdiv N_k \) for all \( i \in \bar T \) while \( x - \theta^k(u^\alpha) \) does not right divide \( N_k \). This means that \( k \notin \bar T \) and therefore \( k \) was already an element of \( \bar T^c \). Thus, the degree of \( Q_{T^c} \) is \( 0 \) when \( T^c \) is empty and is always increased by \( 1 \) when a new element in \( \{ 0, \dots, n-1 \} \) is taken from \( T \) to \( T^c \), so \( \deg(Q_{T^c}) = |T^c| \) and \eqref{eq:TcT} is true for any cardinality of \( T^c \) from \( 0 \) to \( n \).
\end{proof}

In the next results, \(\mathcal{S}\), \(\widehat{\mathcal{S}}\) and \(\Theta\) are the ones in \eqref{eq:ThetaS}.

\begin{prop}\label{prop:gammap}
	Let \( \gamma \in L \) be such that \eqref{eq:gamma36ish} holds.
	Then, there exists some \( \bar \alpha \in L \) such that \(\left\{ \theta^i(\bar \alpha) \norm[\theta]{i}(u^{-1}) ~|~ 0 \leq i \leq n-1 \right\}\) is a \(K\)-basis for \(L\) and \( \gamma' = \theta(\gamma)^{-1} \) is equal to \( (u^{-1})^{\bar \alpha} \). In particular, skew BCH \( \lambda^{-1} \)-constacyclic codes, as described in Definition \ref{defnBCHlike}, can be constructed from \( \gamma' \).
\end{prop}
\begin{proof}
	By applying Lemma \ref{lemma:36ish} with an empty \( T \) and \( T^c = \{ 0, \dots, n-1 \} \), \( \gamma \) is such that \( h = \lcrm{x - \theta^i(\gamma) ~|~ 0 \le i \le n-1} \) equals \( x^n - \lambda \). Thus,
	\( h \mathcal S \), and therefore \( \Theta\left( h \mathcal S\right) = \widehat{\mathcal S} \Theta(h) \), only contain the zero element in \( \mathcal S \) and \( \widehat{\mathcal S} \) respectively. By Lemma \ref{lemma:35ish}, \( \widehat{\mathcal S} \Theta(h) = \widehat{\mathcal S} \hat h \) where \( \hat h = \lclm{x - \theta^i(\gamma') ~|~ 0 \le i \le n-1} \). The code \( \widehat{\mathcal S} \hat h \) being the zero code means that \( \hat h = x^n - \lambda^{-1} \), which in turn implies that \( 0 = \hat h(\gamma') = \norm[\theta]{n}(\gamma') - \lambda^{-1} \), so \( \norm[\theta]{n}(\gamma') = \lambda^{-1} \).

	As \( \norm[\theta]{n}(u \gamma') = \norm[\theta]{n}(u)\norm[\theta]{n}(\gamma') = \lambda \lambda^{-1} = 1 \), by Hilbert's 90 Theorem (see \cite[Chapter VI, Theorem 6.1]{Lang02}) there exists some nonzero \( \bar \alpha \in L \) such that \( 1^{\bar \alpha} = u \gamma' \) and therefore \( \gamma' = (u^{-1})^{\bar \alpha} \). Since \( \deg(\hat h) = n \), the set \( \{ \theta^i((u^{-1})^{\bar\alpha}) ~|~ 0 \le i \le n-1 \} \), which applying \eqref{theta^i(u^alpha)} can be written as \( \{ (u^{-1})^{\theta^i(\bar\alpha) \norm[\theta]{i}(u^{-1})} ~|~ 0 \le i \le n-1 \} \), is P-independent. Hence, \cite[Theorem 5.3]{DL07} applies and \( \{ \theta^i(\bar\alpha) \norm[\theta]{i}(u^{-1}) ~|~ 0 \le i \le n-1 \} \) is a \( K \)-basis of \( L \).
\end{proof}

We can now prove the following generalization of \cite[Theorem 37]{GLN19t}, which adds additional detail to Proposition \ref{prop:dual} and Proposition \ref{prop:dual2}.

\begin{prop}
	Let \( \gamma, T, T^c \) be as in Lemma \ref{lemma:36ish} and let \( g = \lclm{x - \theta^i(u^\alpha) ~|~ i \in T} \).
	Then, \( (\mathcal S g)^\perp = \widehat{\mathcal S} \hat h \) where \( \widehat{\mathcal S} \) is as in \eqref{eq:ThetaS} and
	\[ \hat h = \lclm{x - \theta^i(\gamma') ~|~ i \in T^c}, \]
	where \( \gamma' = \theta(\gamma)^{-1} \). In addition, if \( g \in R \), then \( \hat h \in R \) and \( (\mathcal R g)^\perp = \widehat{\mathcal R} \hat h \) where \( \widehat{\mathcal R} \) is as in \eqref{eq:hatR}.
\end{prop}
\begin{proof}
	Lemma \ref{lemma:36ish} has shown that \( h g = x^n - \lambda \) where \( h = \lcrm{x - \theta^i(\gamma) ~|~ i \in T^c} \). Hence, this follows from Proposition \ref{prop:dual2}.
\end{proof}

\begin{prop}\label{prop:37ish}
	Consider a skew BCH \( \lambda \)-constacyclic code \( \mathcal C = \mathcal R g \) where \( g \) is as shown in Definition \ref{defnBCHlike}. If \( \gamma \in L \) is such that
	\[ (x - \gamma)\lclm{x - \theta(u^\alpha), \dots, x - \theta^{n-1}(u^\alpha)} = x^n - \lambda, \]
	then \( \mathcal C^\perp \) is the skew BCH \( \lambda^{-1} \)-constacyclic code \( \widehat{\mathcal R} \hat h \) where \( \widehat{\mathcal R} \) is as in \eqref{eq:hatR} and
	\[ \hat h = \lclm{x - \theta^{i+j\mu}(\gamma') ~|~ r + \delta - 1 \le i \le r + \mu - 1, 0 \leq j \leq s-1}, \]
	where \( \gamma' = \theta(\gamma)^{-1} \). In particular, it is an \( [ n, s\delta - s] \) code with bounded minimum Hamming distance \(\distance(\mathcal{C}^\perp) \geq \mu - \delta + 2 \).
\end{prop}
\begin{proof}
	Let \( T = \{ i + j \mu \mod n ~|~ r \le i \le \delta - 2, 0 \le j \le s-1 \} \) and \( T^c = \{ i + j \mu \mod n ~|~ r + \delta - 1 \le i \le r + \mu - 1, 0 \le j \le s-1 \} \). Then \( g = \lclm{x - \theta^i(u^\alpha) ~|~ i \in T} \) and \( \hat h = \lclm{x - \theta^i(\gamma') ~|~ i \in T^c} \). The previous result shows that \(  \hat h \in R \) and \( \mathcal C^\perp = \widehat{\mathcal R}\hat h \), which has the asserted properties by virtue of Proposition \ref{propBCHlike}, which can be applied to \( \hat h \) as a result of Proposition \ref{prop:gammap}.
\end{proof}

By choosing the set \( T \) to suit our definition of skew BCH constacyclic codes, we get the following result.

\begin{theorem}\label{thm:BCHLCP}
Given a skew BCH \(\lambda\)-constacyclic \([n,n-s\delta+s]\)-code \(\mathcal{C} = \mathcal{R}g\) of designed distance \(\delta\) where
\[
g = \lclm{x - \theta^{i+j\mu}(u^\alpha) ~|~ r \leq i \leq r + \delta - 2, 0 \leq j \leq s-1}
\]
for any \(0 \leq r \leq n-1\), the code \(\mathcal{D} = \mathcal{R} h\) where
\[
h = \lclm{x - \theta^{i+j\mu}(u^\alpha) ~|~ r + \delta-1 \leq i \leq r + \mu - 1, 0 \leq j \leq s-1}
\]
satisfies that \((\mathcal{C},\mathcal{D})\) is a linear complementary pair of skew BCH \(\lambda\)-constacyclic codes with security parameter \(\delta\).
\end{theorem}

\begin{proof}
It is a direct consequence of Theorem \ref{thm:LCP}, \eqref{x^n-lambda_decomposition}, Proposition \ref{propBCHlike} and Proposition \ref{prop:37ish}.
\end{proof}

\section{Isomorphisms}\label{section:isomorphisms}

In this section we study a subgroup of the group of Hamming distance-preserving inner automorphisms in the ring $\mathcal S$, defined in \eqref{introS}, that can be restricted to automorphisms in $\mathcal R$, as well as Hamming distance-preserving isomorphisms between skew constacyclic codes with a different constant \( \lambda \).
In order to do so, we introduce some notation. Given \(u \in \units{\mathcal{S}}\), we write as \(\phi_u : \mathcal{S} \to \mathcal{S}\) the inner automorphism \(\phi_u(g) = u g u^{-1}\), so the group of inner automorphisms of \(\mathcal{S}\) is
\(
\Inn(\mathcal{S}) = \left\{ \phi_u ~:~ u \in \units{\mathcal{S}} \right\}.
\)

For each \(\beta \in \units{L}\), let \(\overline{\varphi_\beta} : S \to S\) be defined on monomials as \(\overline{\varphi_\beta}(ax^i) = \norm[\theta]{i}(\beta)^{-1} a x^i\) and extended linearly to \(S\). Since \(\overline{\varphi_\beta}\) is clearly bijective, its inverse being \( \overline{\varphi_{\beta^{-1}}} \), and
\[
\overline{\varphi_\beta}(x) \overline{\varphi_\beta}(a) = \beta^{-1} x a = \theta(a) \beta^{-1} x = \overline{\varphi_\beta}(\theta(a)) \overline{\varphi_\beta}(x),
\]
it follows from \cite[Lemma 1.11]{Goodearl/Warfield:2004} that \(\overline{\varphi_\beta}\) is a ring automorphism. As
\begin{equation}\label{eq:baroncentral}
\begin{split}
\overline{\varphi_\beta}(x^n - \lambda) &= \norm[\theta]{n}(\beta^{-1})x^n - \lambda \\
&= \norm{L/K}(\beta^{-1})x^n - \lambda \\
&= \norm{L/K}(\beta^{-1})(x^n - \norm{L/K}(\beta)\lambda),
\end{split}
\end{equation}
the automorphism $\overline{\varphi_\beta}$ sends the two-sided ideal $S(x^n - \lambda)$ into the two-sided ideal $S(x^n - \norm{L/K}(\beta)\lambda)$, and therefore can be projected into the ring isomorphism
\[
	\varphi_\beta : \mathcal S \to \mathcal S' =  \frac{S}{S(x^n - \norm{L/K}(\beta)\lambda)}
\]
such that, for any monomial \( a x^i \in \mathcal S \),
\begin{equation}\label{eq:phbmonomials}
	\varphi_\beta(a x^i) = \norm[\theta]{i}(\beta^{-1}) a x^i.
\end{equation}

For each \(\alpha \in \units{L}\), let \(\phi_\alpha \in \Inn(\mathcal{S})\), i.e.,
\[
\phi_\alpha: \mathcal{S} \to \mathcal{S}, \ \phi_\alpha(g) = \alpha g \alpha^{-1}.
\]

\begin{prop}\label{varphibetainner}
For any \( \beta \in \units L \) such that \(\norm{L/K}(\beta) = 1\), there is some \(\alpha \in \units L\) such that \(\varphi_\beta = \phi_\alpha\). Conversely, for each \(\alpha \in \units{L}\), \(\phi_\alpha = \varphi_{\alpha^{-1}\theta(\alpha)}\).
\end{prop}

\begin{proof}
Since \(\norm{L/K}(\beta) = 1\), it follows that \(\mathcal{S}' = \mathcal{S}\). Hence, \(\varphi_\beta\) is an automorphism. In addition, there exists some \( \alpha \in \units L \) such that \(\beta = \alpha^{-1}\theta(\alpha)\) by Hilbert's Theorem 90 (see \cite[Chapter VI, Theorem 6.1]{Lang02}). As \(\norm[\theta]{i}(\beta^{-1}) = \alpha \theta^i(\alpha^{-1})\), we have
\begin{equation}\label{eq:phbmonomialsauto}
\varphi_\beta(a x^i) = \norm[\theta]{i}(\beta^{-1}) a x^i = \alpha \theta^i(\alpha^{-1}) a x^i = \alpha a x^i \alpha^{-1},
\end{equation}
which implies \(\varphi_\beta = \phi_\alpha\). The converse also follows straightforwardly from \eqref{eq:phbmonomialsauto}.
\end{proof}

\begin{prop}
The map
\[
\begin{split}
\{ \beta \in \units{L} : \norm{L/K}(\beta) = 1 \} &\to \Inn(\mathcal{S}) \\
\beta &\mapsto \varphi_\beta
\end{split}
\]
is an injective morphism of groups.
\end{prop}

\begin{proof}
Since the norm is multiplicative, it follows that \(N = \{ \beta \in \units{L} : \norm{L/K}(\beta) = 1 \}\) is a subgroup of \(\units{L}\). Let \(\beta_1, \beta_2 \in N\). It is immediate from \eqref{eq:phbmonomials} that \(\varphi_{\beta_1 \beta_2} = \varphi_{\beta_1} \circ \varphi_{\beta_2}\). Moreover, if \(\varphi_{\beta_1} = \varphi_{\beta_2}\) then
\[
\beta_1^{-1} x = \varphi_{\beta_1}(x) = \varphi_{\beta_2}(x) = \beta_2^{-1} x,
\]
hence \(\beta_1 = \beta_2\).
\end{proof}

Observe that $(\lambda^{-1}x^{n-1})x = \lambda^{-1}x^n = 1$, hence \(x \in \units{\mathcal{S}}\) and $x^{-1} = \lambda^{-1}x^{n-1}$.
Therefore, we may also consider the inner automorphism $\phi_x : \mathcal S \to \mathcal S$ as $\phi(g) = xgx^{-1}$ for any $g \in \mathcal S$. On monomials,
\begin{equation}
\label{eq:phxmonomials}
\phi_x(ax^i) = xax^ix^{-1} = \theta(a)x^i.
\end{equation}

The automorphism \( \phi_x \) can similarly be seen as the projection into \( \mathcal S \) of the map \mbox{\( \overline{\phi_x} : S \to S \)} defined over monomials as \( \overline{\phi_x}(a x^i) = \theta(a)x^i \), as \( \overline{\phi_x} \) fixes any skew polynomial in \( Z(S) \), including \( x^n - \lambda \). An analogous reasoning as the one shown above for \( \overline{\varphi_\beta} \) proves that it is an automorphism.

\begin{prop}\label{prop:isometries}
The inner automorphism \(\phi_x \in \Inn(\mathcal S) \) and the isomorphisms \( \varphi_\beta : \mathcal S \to \mathcal S' \) where \( \beta \in \units L \), which are also inner automorphisms whenever \(\norm{L/K}(\beta) = 1 \), as well as their restrictions to \( \mathcal R \), preserve the Hamming weight, i.e., they are isometries with respect to the Hamming distance.
\end{prop}

\begin{proof}
They are isometries as a consequence of \eqref{eq:phbmonomials} and \eqref{eq:phxmonomials}. If \(\norm{L/K}(\beta) = 1 \), \( \varphi_\beta \) is inner by Proposition \ref{varphibetainner}.
\end{proof}

Moreover, when restricting these maps to \( \mathcal R \subseteq \mathcal S \), they may result in isomorphisms between rings whose left ideals are skew constacyclic codes, and therefore, by Proposition \ref{prop:isometries}, establish isometries between skew constacyclic codes, as we shall now study.

\begin{lemma}\label{lemma:phitocalR}
The restriction of \( \overline{\phi_x} \) to \( R \) is also an automorphism.
The restriction of \(\phi_x\) to \(\mathcal{R}\) is also an inner automorphism.
\end{lemma}

\begin{proof}
	This is immediate from \eqref{eq:phxmonomials} and the fact that \(\theta(a) = \sigma(a)\) for all \(a \in F\).
\end{proof}

The following lemma applies for any $L/F$ and $F/K$ finite, Galois extensions for some fields $L,F,K$.

\begin{lemma}\label{lemma:NLFK}
	For any $a \in F$, $\norm{L/K}(a) = \norm{F/K}(a)^{[L:F]}$.
\end{lemma}
\begin{proof}
	$\norm{L/K} = \norm{F/K} \circ \norm{L/F}$, see \cite[Chapter VI, Theorem 5.1]{Lang02}. For any $a \in F$, $\norm{L/F}(a)$ is equal to $a^{[L:F]}$ by definition of the norm, as $\pi(a) = a$ for any $\pi \in \Gal(L/F)$ and $|\Gal(L/F)| = [L:F]$. As the norm map is multiplicative, $\norm{L/K}(a) = \norm{F/K}\left(a^{[L:F]}\right) = \norm{F/K}(a)^{[L:F]}$.
\end{proof}

In our context, where \( L/F \) is a field extension from Definition \ref{degreesextension}, \( [L:F] = s \). We will heavily use the fact that \( \norm{L/K}(a) = \norm{F/K}(a)^s \) from now on.

\begin{lemma}\label{lemma:phibtocalR}
If \(\beta \in \units{F}\), then the restriction of \( \overline{\varphi_\beta} \) to \( R \) is also an automorphism, while \(\varphi_\beta\) restricts to an isomorphism \(\varphi_\beta : \mathcal{R} \to \mathcal{R}'\), where
\[
\mathcal R' = \frac{R}{R(x^n - \norm{F/K}(\beta)^s \lambda)}.
\]
\end{lemma}

\begin{proof}
	\(\norm{L/K}(\beta) = \norm{F/K}(\beta)^s\) by Lemma \ref{lemma:NLFK}, and \(\norm[\theta]{i}(\beta) = \norm[\sigma]{i}(\beta)\). Hence, \(\overline{\varphi_\beta}(R) \subseteq R\) and \(\overline{\varphi_\beta}(x^n-\lambda) = \norm{F/K}(\beta)^{-s}(x^n - \norm{F/K}(\beta)^s \lambda)\) by \eqref{eq:baroncentral}. The lemma follows.
\end{proof}

It is immediate from this that skew \( \lambda \)-constacyclic codes are isometrically isomorphic to skew cyclic codes if there is some \( u \in F \) such that \( \norm{F/K}(u)^s = \lambda \), that is, if the latest statement in Proposition \ref{prop:decomposition} is true for some \( u \in F \). The family of isomorphisms given by Lemma \ref{lemma:phibtocalR} has been described for skew constacyclic codes over finite fields in \cite[Section 3]{BBB21}, where it is used to study whether skew \( \lambda \)-constacyclic codes, skew cyclic codes and skew negacyclic codes are equivalent. If \( s = 1 \), it follows that any \( u \in L \) satisfying the latest statement in Proposition \ref{prop:decomposition} is also in \( F \) as \( L = F \). See also \cite[Proposition 23]{GLN19b} for an isomorphism of this family for the case where \( s = 1 \), \(\mathcal R' = R/R(x^n - 1)\) and \( F \) is of the form \( \FF_q(t) \).

\begin{lemma}\label{lemma:phibtocalRautomorphism}
Let \(\beta \in \units{L}\) such that \(\norm{L/K}(\beta) = 1\). Then \(\varphi_\beta\) restricts to an automorphism of \(\mathcal{R}\) if and only if \(\beta \in \units{F}\).
\end{lemma}

\begin{proof}
Since $\varphi_\beta(x) = \beta^{-1} x$, if \(\varphi_\beta\) restricts to an automorphism of \(\mathcal{R}\) then \(\beta \in \units{F}\). The reciprocal statement follows from the fact that \(\norm{F/K}(\beta)^s = \norm{L/K}(\beta) = 1\) as a result of Lemma \ref{lemma:phibtocalR}.
\end{proof}

\begin{remark}
As shown in the previous proof, \( \norm{L/K}(\beta) = 1 \) implies that \(\norm{F/K}(\beta)\) is an \(s\)-th root of unity, not necessarily being \( 1 \). Hence, there does not necessarily exist any \(\alpha \in F\) such that \(\beta = \sigma(\alpha)\alpha^{-1}\), so the restriction of the inner automorphism \( \varphi_\beta : \mathcal S \to \mathcal S \) to an automorphism in \( \mathcal R \) is not necessarily inner.
\end{remark}

\newcommand{\R}[2]{\mathcal R_{#1, #2}}
These isomorphisms also have the following property, which results in them also keeping the dual distance (that is, the Hamming distance of its dual) of any code. For convenience, we shall define \( \R n a \) as \( \frac R {R(x^n - a)} \) for any \( a \in K = F^\sigma \) and any \( n \in \ZZ^+ \). Note that the left ideals of \( \R n a \) are skew \( a \)-constacyclic codes of length \( n \), and \( \mathcal R = \R n \lambda \).

\begin{prop}\label{prop:isomdual}
For any \(\beta \in \units F\), any \( i \in \ZZ \) and any left ideal $\mathcal D \subseteq \mathcal R$, it follows that \( (\varphi_\beta \circ \phi_x^i)(\mathcal D)^\perp \) and \( (\varphi_{\beta^{-1}} \circ \phi_x^i)(\mathcal D^\perp) \) are the same left ideal in the ring \( \R n {\norm{F/K}(\beta)^{-s}\lambda^{-1}} \).
\end{prop}

\begin{proof}
Consider any left ideal \( \mathcal C \subseteq \R n a \). By Proposition \ref{prop:dual}, \(\mathcal{C}^\perp \) is a left ideal in \( \R n {a^{-1}} \), while Lemmas \ref{lemma:phitocalR} and \ref{lemma:phibtocalR} show, respectively, that \( \phi_x(\mathcal C) \) is a left ideal in \( \R n a \) and \( \varphi_\beta(\mathcal C) \) is a left ideal in \( \R n {\norm{F/K}(\beta)^s a} \). This means that both \( (\varphi_\beta \circ \phi_x^i)(\mathcal D)^\perp \) and \( (\varphi_{\beta^{-1}} \circ \phi_x^i)(\mathcal D^\perp) \) are left ideals in \( \R n {\norm{F/K}(\beta)^{-s}\lambda^{-1}} \).

The vector space isomorphism \eqref{eq:vsisom} can be considered for any ring of the form \( \R n a \). As a result, we may evaluate the inner product for any pair of elements in two rings of the form \( \R n a, \R n b \): for any \(f = \sum_{i=0}^{n-1} a_i x^i \in \R n a \) and any \( g = \sum_{i=0}^{n-1} b_i x^i \in \R n b\) such that \( a, b \in K \), the inner product of $f$ and $g$ is defined as
\[
\pesc{f}{g} = \sum_{i=0}^{n-1} a_i b_i.
\]

As $\sigma$ is a ring automorphism in $F$, we get from \eqref{eq:phxmonomials} that
\[
\pesc {\phi_x(f)} {\phi_x(g)} = \sum_{i=0}^{n-1} \sigma(a_i) \sigma(b_i) =   \sigma(\pesc f g)
\]
since \(\theta_{\mid F} = \sigma\). 
Thus, $\pesc {\phi_x(f)} g = 0$ if and only if $\pesc f {\phi_x^{-1}(g)} = 0$. Therefore, if \( \mathcal C \subseteq \R n a \), then
\[
\begin{aligned}
		\phi_x(\mathcal C)^\perp & = \{ g \in \R n {a^{-1}} : \pesc f g = 0 \ \tforall f \in \phi_x(\mathcal C) \} \\
		& = \{ g \in \R n {a^{-1}} : \pesc {\phi_x(f)} g = 0 \ \tforall f \in \mathcal C \} \\
		& = \{ g \in \R n {a^{-1}} : \pesc f {\phi_x^{-1}(g)}= 0 \ \tforall f \in \mathcal C \} \\
		& = \{ \phi_x(g) : g \in \R n {a^{-1}}, \pesc f g = 0 \ \tforall f \in \mathcal C \} \\
		& = \phi_x(\mathcal C^\perp).
	\end{aligned}
	\]
On the other hand, by \eqref{eq:phbmonomials}
\[
\pesc f {\varphi_\beta(g)} = \sum_{i=0}^{n-1} \norm[\theta]{i}(\beta)^{-1} a_i b_i = \pesc {\varphi_\beta(f)} g,
\]
so for any \( \mathcal C \subseteq \R n a \)
\[
\begin{aligned}
		\varphi_\beta(\mathcal C)^\perp & = \{ g \in \R n {\norm{F/K}(\beta)^{-s}a^{-1}} : \pesc f g = 0 \ \tforall f \in \varphi_\beta(\mathcal C) \} \\
		& = \{ g \in \R n {\norm{F/K}(\beta)^{-s}a^{-1}} : \pesc {\varphi_\beta(f)} g = 0 \ \tforall f \in \mathcal C \} \\
		& = \{ g \in \R n {\norm{F/K}(\beta)^{-s}a^{-1}} : \pesc f {\varphi_\beta(g)}= 0 \ \tforall f \in \mathcal C \} \\
		& = \{ \varphi_{\beta^{-1}}(g) : g \in \R n {a^{-1}}, \pesc f g = 0 \ \tforall f \in \mathcal C \} \\
		& = \varphi_{\beta^{-1}}(\mathcal C^\perp).
\end{aligned}
\]
By combining the previous identities and applying them for a left ideal \( \mathcal D \subseteq \mathcal R = \R n \lambda \), we get that
\[
	\varphi_\beta(\phi_x^i(\mathcal D))^\perp
	= \varphi_{\beta^{-1}}(\phi_x^i(\mathcal D)^\perp)
	= \varphi_{\beta^{-1}} (\phi_x^i(\mathcal D^\perp)),
\]
which is the asserted equality of the two left ideals.
\end{proof}

These results mean that we may replace one or both of the codes in a pair \( (\mathcal C, \mathcal D) \) with the result of applying these isomorphisms to them and the security parameter remains the same.

\newcommand{\isom}{\zeta}
\begin{prop}\label{prop:isom-secparam}
	The security parameter, as given in Definition \ref{defn:secparam}, for a pair of codes \( (\mathcal C, \mathcal D) \) is the same as the one for \( (\isom_1(\mathcal C), \isom_2(\mathcal D)) \) where \( \isom_1 = \varphi_{\beta_1} \circ \phi_x^{i_1} \) and \( \isom_2 = \varphi_{\beta_2} \circ \phi_x^{i_2} \) for any \( \beta_1, \beta_2 \in \units F \) and any \( i_1, i_2 \in \ZZ \).
\end{prop}

\begin{proof}
	Proposition \ref{prop:isometries} shows that \( \distance(\mathcal C) = \distance(\isom_1(\mathcal C)) \) as the composition of isometries is an isometry. By Proposition \ref{prop:isomdual}, \( \isom_2(\mathcal D)^\perp = (\varphi_{\beta_2} \circ \phi_x^{i_2})(\mathcal D)^\perp = (\varphi_{(\beta_2)^{-1}} \circ \phi_x^{i_2})(\mathcal D^\perp) \), which has the same minimum distance as \( \mathcal D^\perp \) for the same reason.
\end{proof}

Lemmas \ref{lemma:phitocalR} and \ref{lemma:phibtocalRautomorphism} have shown which ones of the previously discussed isomorphisms are automorphisms. These span a group of automorphisms in \( \mathcal R \), which can be used to modify a pair of skew \( \lambda \)-constacyclic codes into another one with the same security parameter. The following result explicitly describes this group.

\begin{prop}\label{prop:automorphismgroup}
	The group \(G\) spanned by $\phi_x$ and $\{ \varphi_\beta : \beta \in F, \norm{F/K}(\beta)^s = 1 \}$ is
	\begin{equation}\label{eq:G}
		G = \{ \varphi_\beta \circ \phi_x^i : \beta \in F, \ \norm{F/K}(\beta)^s = 1, \ 0 \le i < \mu \}
	\end{equation}
	where $\mu = |\sigma|$.
	These elements are all distinct. Therefore, if $F$ is a finite field, $G$ has exactly
	\begin{equation}\label{eq:orderG}
		\mu\frac{|\units F|}{|\units K|}\gcd(s, |\units K|)
	\end{equation}
	 elements.
\end{prop}

\begin{proof}
Let \(H = \{ \varphi_\beta \circ \phi_x^i : \beta \in F, \ \norm{F/K}(\beta)^s = 1, \ 0 \le i < \mu \}\). Clearly, \(H \subseteq G\). Since \(\varphi_\beta \in H\) for each \(\beta \in F\) such that \(\norm{F/K}(\beta)^s = 1\) and \( \phi_x \in H \), we have to prove that \(H\) is closed under composition and inverses in order to prove \(G = H\). For any $0 \le i, j < \mu$ and any $\beta_1,\beta_2 \in F$ such that $\norm{F/K}(\beta_1)^s = \norm{F/K}(\beta_2)^s = 1$, it is straightforward to check that
\[
\norm{F/K}(\beta_1 \sigma^i(\beta_2))^s = 1
\]
by the multiplicativity of the norm, and
\[
\varphi_{\beta_1} \circ \phi_x^i \circ \varphi_{\beta_2} \circ \phi_x^j = \varphi_{\beta_1 \sigma^i(\beta_2)} \circ \phi_x^{i+j \bmod{n}}.
\]
It follows that \(H\) is closed under composition and that the inverse map of $\varphi_{\beta} \circ \phi^i$ is $\varphi_{\sigma^{-i}(\beta^{-1})} \circ \phi^{\mu-i}$, which is also in $H$.

	In order to show that these elements are all distinct, it is enough to check that the only element equal to the identity map is $\varphi_1 \circ \phi^0$. If $\varphi_\beta \circ \phi_x^i$ is the identity map, then for any $\alpha \in F$ such that $\sigma^i(\alpha) \ne \alpha$ for all $1 \le i < \mu$, we get from \eqref{eq:phbmonomials} and \eqref{eq:phxmonomials} that
\[
\alpha = (\varphi_\beta \circ \phi_x^i)(\alpha) = \varphi_\beta(\sigma^i(\alpha)) = \norm[\theta]{0}(\beta)^{-1} \sigma^i(\alpha) = \sigma^i(\alpha),
\]
so $i = 0$. Again by \eqref{eq:phbmonomials} and \eqref{eq:phxmonomials},
\[
x = (\varphi_\beta \circ \phi_x^i)(x) = \varphi_\beta(x) = \norm[\theta]{1}(\beta)^{-1}x = \beta^{-1}x,
\]
which implies that $\beta = 1$.

Finally, if \(F\) is finite, there are exactly \(\frac{|\units{F}|}{|\units{K}|}\) elements in \(\units{F}\) with a given nonzero norm with respect to the extension \( F/K \), and there are \(\gcd(s,|\units{K}|)\) \(s\)-th roots of \(1\) in the cyclic group \( \units K \). As a result of these well-known facts (see \cite{Lidl97}), the order of \(G\) is \(\mu \frac{|\units{F}|}{|\units{K}|}\gcd(s,|\units{K}|)\). 
\end{proof}

\section{Examples of constructions of LCPs of skew constacyclic codes with control over the security parameter}\label{section:examples}

Theorem \ref{thm:LCP} gives a simple criterion for evaluating whether, given the generator polynomials for two skew \( \lambda \)-constacyclic codes \( \mathcal C, \mathcal D \subset \mathcal R \) where \( \mathcal R \) is as in \eqref{introR}, the pair $(\mathcal C, \mathcal D)$ constitutes a linear complementary pair of codes. On the other hand, computing the security parameter for the pair, or a lower bound thereof, requires, in principle, either prior knowledge about both \( \mathcal C \) and \( \mathcal D^\perp \) or running some algorithms which are significantly more costly than the one for determining whether the pair is an LCP of codes. The same is true for the more general setting considered in \cite{Ngo15}.

However, the isomorphisms described in the previous section, as well as the results on duals of skew constacyclic codes, allow us to take different approaches from the one in \cite{Ngo15}. We may first consider a pair \( (\mathcal C, \mathcal D) \) of skew \( \lambda \)-constacyclic codes of length \( n \) such that the sum of their dimensions is \( n \) and the security parameter is high enough for the desired purposes, not necessarily caring at first about whether the pair is an LCP of codes. Then, we may apply automorphisms in the group \( G \), described in \eqref{eq:G}, to one element in the pair until the modified pair, which keeps the same security parameter by Proposition \ref{prop:isom-secparam}, is shown to be an LCP of skew constacyclic codes through Theorem \ref{thm:LCP}.

Theorem \ref{thm:BCHLCP}, which directly yields LCPs of codes, offers a relatively straightforward way to generate a pair of codes with some known lower bound for its security parameter. In fact, for \( s = 1 \), it yields LCPs of MDS codes, so the resulting security parameter is the maximum possible for an \( [n, k] \) code and an \( [n, n-k] \) code, which is the value from the Singleton bound, \( n - k + 1 \). However, one might be interested in another approach if the minimum distance given by BCH codes happens to be unsatisfactory for the required purposes, which might be the case for \( s > 1 \), or if one or both of the codes are required to belong to a particular subfamily of skew constacyclic codes in order to, for example, support a particular error-correcting decoding algorithm. As we are considering a security parameter as high as possible and this is covered by Theorem \ref{thm:BCHLCP} for \( s = 1 \), we will restrict ourselves to \( s > 1 \) in the following examples, although analogous constructions can be done for \( s = 1 \).

If exacly one of the codes, either \( \mathcal C \) or \( \mathcal D \), is given as a result of some restrictions, then the security parameter is at most \( \distance(\mathcal C) \) or, respectively, \( \distance(\mathcal D^\perp) \). In this case, we may attempt to complete the pair in such a way that the security parameter takes exactly that value.
There are several options for doing so. We will illustrate a few of them through examples.

As the dual of a skew \( \lambda \)-constacyclic code is a skew \( \lambda^{-1} \)-constacyclic code, as shown in Proposition \ref{prop:dual}, if \( \lambda^{-1} = \lambda \) then the most immediate option is to take \( \mathcal D = \mathcal C^\perp \). This is equivalent to checking whether the given code is a linear complementary dual code; that is, if \( (\mathcal C, \mathcal C^\perp) \) or \( (\mathcal D^\perp, \mathcal D) \) is an LCP of codes. If it is not an LCP of codes, such a pair may result by applying isomorphisms to one of the members in the pair.

\begin{example}\label{ex:1roos}
	Take as \( \mathcal C = \mathcal R g \) the first code in \cite[Table 1]{ALN21}, which is the skew BCH cyclic \( [12, 6, 6] \) \( F \)-linear code for \( L = \FF_{2^{12}} = \FF_2(\gamma) \) where \( \gamma^{12} + \gamma^7 + \gamma^6 + \gamma^5 + \gamma^3 + \gamma + 1 = 0 \), \( \theta : L \to L \) is the Frobenius endomorphism (hence \( K = L^\theta = \FF_2, n = 12 \)), \( F = L^{\theta^6} = \FF_{2^6} \) (so \( \mu = 6, s = 2, \sigma = \theta_{|F} \)), \( \alpha = \gamma^5 \) and
	\[
		g = \lclm{x - \theta^{i + 6j}(\alpha^{-1}\theta(\alpha)) ~|~ 2 \le i \le 4, 0 \le j \le 1}.
	\]

	Note that $\alpha^{-1} \theta(\alpha) = u^\alpha$ for $u = 1$, using the notation in \eqref{eq:conjugate}. By Proposition \ref{propBCHlike}, \( g \in R \) and the minimum Hamming distance of \( \mathcal C = \mathcal R g \) is at least \( 4 \), while according to \cite[Table 1]{ALN21} it is in fact \( 6 \). Such a minimum distance cannot be guaranteed for \( \mathcal D^\perp \), which should have dimension \( 6 \), by constructing it through Theorem \ref{thm:BCHLCP} or Proposition \ref{propBCHlike}.
	
	\newcommand{\genF}{\eta}
	If we choose \( \genF = \gamma^9 + \gamma^5 + \gamma^4 + \gamma^2 + \gamma \) as the element such that \( F = K(\genF) \), then \( g = x^6 + (\genF^4 + \genF^3 + 1) x^5 + (\genF^4 + \genF^3 + \genF) x^4 + \genF x^3 + \genF^5 x^2 + (\genF^4 + \genF^3 + \genF^2 + \genF) x + \genF^2 + 1 \) and we may compute the generator \( h^\Theta \) of \( (\mathcal R g)^\perp \) from \( h \), where \( hg = x^{12} - 1 \), using Equation \eqref{eq:dualgenerator} and Definition \ref{defn:monicrec}. This results in \( h^\Theta = x^6 + (\genF^5 + \genF^4)x^5 + (\genF^4 + \genF^3 + \genF^2)x^4 + (\genF^5 + 1)x^3 + (\genF^5 + \genF^3 + 1)x^2 + (\genF^5 + \genF^3 + \genF^2 + \genF + 1)x + \genF^2 + 1 \). The greatest common right divisor of \( g \) and \( h^\Theta \) is \( x^2 + (\genF^5 + \genF^2)x + \genF^5 + \genF^3 + \genF^2 + 1 \). As a result, we have to use the previously studied automorphisms to produce an LCP of codes. For example, \( \phi_x^2(h^\Theta) = x^6 + (\genF^4 + \genF^2)x^5 + (\genF^3 + \genF^2 + 1)x^4 + (\genF^2 + \genF)x^3 + (\genF^5 + \genF^3 + \genF^2 + \genF + 1)x^2 + \genF^3x + \genF^5 + \genF^4 + \genF^2 + \genF \) is such that \( \gcrd{g, \phi_x^2(h^\Theta)} = 1 \). Hence, by Theorem \ref{thm:LCP}, \( (\mathcal R g, \mathcal R \phi_x^2(h^\Theta)) \) is an LCP of codes with security parameter \( 6 \), which is as high as possible for the chosen \( \mathcal C = \mathcal R g \). We can also get a different LCP of codes with \( \phi_x^4(h^\Theta) \), but not \( \phi_x(h^\Theta) \), \( \phi_x^3(h^\Theta) \) or \( \phi_x^5(h^\Theta) \) as they share a proper right divisor with \( g \).

	For every \( u \in \units F \) it is true that \( \norm{F/K}(u)^s = 1 \), as the norm of any element is nonzero and in \( K \), which has \( 1 \) as its only nonzero element. As a result, \( 144 \) LCPs of codes with security parameter \( 6 \) can be obtained from modifying \( h^\Theta \) in \( (g, h^\Theta) \) through the \( 378 \) possible compositions of the \( |\units F| = 63 \) automorphisms of the form \( \varphi_u \) with the six different powers of \( \phi_x \).
\end{example}

\begin{example}\label{ex:9bgu}
	Consider the only code in \cite[Table 2]{BGU07}, which is the skew cyclic \( [44, 20, 17] \) \( \FF_9 \)-linear code generated by \( g = x^{24} + x^{21} + x^{20} + \alpha^7 x^{19} + \alpha^3 x^{18} + 2 x^{17} + \alpha^3 x^{16} + \alpha^5 x^{14} + \alpha^5 x^{13} + 2 x^{12} + \alpha^2 x^{10} + \alpha^7 x^{9} + 2 x^{6} + \alpha^5 x^{5} + \alpha^7 x^{4} + \alpha^3 x^{3} + \alpha^7 x^{2} + \alpha^2 x + 2 \), where \( \alpha \) is such that \( F = \FF_9 = \FF_3(\alpha) \) (we are using \( \alpha^2 = \alpha+1 \)) and \( \sigma \) is the Frobenius endomorphism in \( F \). We shall remark that, according to \cite{Grassltables}, no \( [44, 20] \) code over \( \FF_9 \) is known to have a minimum distance over \( 17 \), which would be required in order to achieve a security parameter over \( 17 \). In fact, the code we are considering was proposed as an improvement over the previously known minimum distance for \( [44, 20] \) codes over \( \FF_9 \).
	
	In this example, we will use \( \mathcal R g \) as both \( \mathcal C \) and \( \mathcal D^\perp \). The generator of \( (\mathcal R g)^\perp \) obtained from Proposition \ref{prop:dual} is \( h = x^{20} + \alpha^6 x^{19} + \alpha x^{18} + \alpha x^{17} + \alpha^2 x^{16} + x^{15} + \alpha^2 x^{14} + \alpha x^{13} + \alpha^2 x^{12} + \alpha^3 x^{11} + \alpha^6 x^{10} + \alpha x^{9} + \alpha^7 x^{8} + \alpha x^{6} + \alpha x^{5} + 2 x^{4} + 2 x^{3} + 1 \). In this case, \( \gcrd{g,h} = 1 \), so \( (\mathcal Rg, \mathcal Rh) \) is a linear complementary pair of codes with security parameter \( 17 \).
	
	By using the isomorphism group described in Proposition \ref{prop:automorphismgroup}, we may get different pairs. In this case, any modification of the pair through applying isomorphisms to one or both of the elements results in yet another linear complementary pair of codes, as the greatest common right divisor is \( 1 \) in all cases. As the cardinality of the isomorphism group is \( 16 \) (\( \mu = 2 \), \( s = 22 \) and \( 2^s = 1 \) so the map \( (\norm{F/K})^s : F \to K \) sends every unit to \( 1 \)), this approach yields \( 16^2 = 256 \) possible LCPs of codes whose security parameter, by Proposition \ref{prop:isom-secparam}, is \( 17 \). In other words, we get \( 16 \) \( [44, 20, 17] \) codes and \( 16 \) \( [44, 24] \) codes with dual distance \( 17 \) such that any pair yields an LCP of codes.
\end{example}

\begin{example}\label{ex:8roos}
	Take as \( \mathcal D = \mathcal R h \) the last code in \cite[Table 1]{ALN21}, which is the skew BCH cyclic \( [10, 4, 7] \) \( F \)-linear code for \( L = \FF_{5^{10}} = \FF_5(\gamma) \) where \( \gamma^{10} + 3\gamma^5 + 3\gamma^4 + 2\gamma^3 + 4\gamma^2 + \gamma + 2 = 0 \), \( \theta : L \to L \) is the Frobenius endomorphism (hence \( K = L^\theta = \FF_5, n = 10 \)), \( F = L^{\theta^5} = \FF_{5^5} \) (so \( \mu = 5, s = 2, \sigma = \theta_{|F} \)), \( \alpha = \gamma^9 + \gamma^7 + \gamma^6 + 3\gamma^5 + 2\gamma^3 + \gamma + 3 = \gamma^{7\,861\,528} \)
	and
	\[
		h = \lclm{x - \theta^{i + 5j}(\alpha^{-1}\theta(\alpha)) ~|~ 0 \le i \le 2, 0 \le j \le 1}.
	\]

	Since \( \mathcal D \) is an MDS code, by \cite[Theorem 2.4.3]{HP03} so is its dual \( \mathcal D^\perp \), which is generated by \( g^\Theta \), where \( g \) is such that \( gh = x^{10} - 1 \), by Proposition \ref{prop:dual}.
	\newcommand{\genF}{\eta}%
	Taking \( \genF = \gamma^{15\,630} \), this results in \( h = x^6 + \genF^{27} x^5 + \genF^{2957} x^4 + \genF^{968} x^3 + \genF^{1148} x^2 + \genF^{2955} x + \genF^{2038} \), and the generator of \( \mathcal D^\perp \) is \( g^\Theta = x^4 + \genF^{599}x^3 + \genF^{1816}x^2 + \genF^{2309}x + \genF^{720} \).
	
	As \( \gcrd{g^\Theta, h} = x + \genF^{570} \), \( (\mathcal D^\perp, \mathcal D) \) fails to be an LCP of codes, while \( (\mathcal R\phi_x(g^\Theta), \mathcal D) \) does result in an LCP of codes with a security parameter of \( \distance(\mathcal D^\perp) = 5 \), which is the maximum possible for a \( [10, 6] \) code and a \( [10, 4] \) code, as \( \phi_x(g^\Theta) = x^4 + \genF^{2995}x^3 + \genF^{2832}x^2 + \genF^{2173}x + \genF^{476} \) does satisfy the condition on the greatest common right divisor. The same is true for \( \phi_x^4(g^\Theta) \), while it is not for \( \phi_x^2(g^\Theta), \phi_x^3(g^\Theta) \).

	By Proposition \ref{prop:automorphismgroup}, there are \( 7810 \) automorphisms available if the ones of the form \( \varphi_\beta \) are also considered. This results in \( 4820 \) possible supplements for \( \mathcal D = \mathcal R h \).
\end{example}

If \( \lambda^{-1} \ne \lambda \), the previous approach can still be used if there is some element \( \beta \in F \) such that \( \norm{F/K}(\beta)^s = \lambda^2 \), as then by Lemma \ref{lemma:phibtocalR} \( \varphi_\beta \) is an isomorphism from \( \widehat{\mathcal R} = R/R(x^n - \lambda^{-1}) \) to \( \mathcal R \), so \( (\mathcal C, \mathcal C^\perp) \) or \( (\mathcal D^\perp, \mathcal D) \) has the same security parameter (which is \( \distance(\mathcal C) \) or, respectively, \( \distance(\mathcal D^\perp) \)) as the pair of skew \( \lambda \)-constacyclic codes \( (\mathcal C, \varphi_\beta(\mathcal C^\perp)) \) or, respectively, \( (\varphi_\beta(\mathcal D^\perp), \mathcal D) \), by Proposition \ref{prop:isom-secparam}.

\begin{example}\label{ex:20,9,10}
	Take as \( \sigma \) the square of the Frobenius endomorphism in \( F = \FF_{2^8} = \FF_2(a) \) where \( a^8 + a^4 + a^3 + a^2 + 1 = 0 \), so \( K = F^\sigma = \FF_{2^2} = \FF_{2}(\lambda) \) for \( \lambda = a^7 + a^6 + a^4 + a^2 + a \), which is such that \( \lambda^2 = \lambda + 1 \). Then define \( R = F[x;\sigma] \) and \( \mathcal R = R/R(x^{20} - \lambda) \). Any left ideal of \( \mathcal R \) is a skew \( \lambda \)-constacyclic code of length \( 20 \). The skew polynomial \( g = x^{11} + a^{101}x^{10} + a^{165}x^9 + a^{157}x^8 + a^{229}x^7 + a^{193}x^6 + a^{211}x^5 + a^{178}x^4 + a^{47}x^3 + a^{112} x^2 + a^{107} x + a^{58} \) right divides \( x^{20} - 1 \), so \( \mathcal C = \mathcal R g \) is a \( [20, 9] \) \( F \)-linear skew \( \lambda \)-constacyclic code.

	A parity-check matrix for the code \( \mathcal C \) is \( \textbf H \), whose transpose is
\[
\textbf H^T =
\left(
\begin{smallmatrix}
1 & 0 & 0 & 0 & 0 & 0 & 0 & 0 & 0 & 0 & 0 \\
0 & 1 & 0 & 0 & 0 & 0 & 0 & 0 & 0 & 0 & 0 \\
0 & 0 & 1 & 0 & 0 & 0 & 0 & 0 & 0 & 0 & 0 \\
0 & 0 & 0 & 1 & 0 & 0 & 0 & 0 & 0 & 0 & 0 \\
0 & 0 & 0 & 0 & 1 & 0 & 0 & 0 & 0 & 0 & 0 \\
0 & 0 & 0 & 0 & 0 & 1 & 0 & 0 & 0 & 0 & 0 \\
0 & 0 & 0 & 0 & 0 & 0 & 1 & 0 & 0 & 0 & 0 \\
0 & 0 & 0 & 0 & 0 & 0 & 0 & 1 & 0 & 0 & 0 \\
0 & 0 & 0 & 0 & 0 & 0 & 0 & 0 & 1 & 0 & 0 \\
0 & 0 & 0 & 0 & 0 & 0 & 0 & 0 & 0 & 1 & 0 \\
0 & 0 & 0 & 0 & 0 & 0 & 0 & 0 & 0 & 0 & 1 \\
a^{58} & a^{107} & a^{112} & a^{47} & a^{178} & a^{211} & a^{193} & a^{229} & a^{157} & a^{165} & a^{101} \\
a^{207} & a^{231} & a^{93} & a^{161} & a^{31} & a^{40} & a^{24} & a^{221} & a^{55} & a^{141} & a^{154} \\
a^{164} & a^{206} & a^{241} & a^{87} & a^{22} & a^{152} & a^{3} & a^{225} & a^{254} & a^{203} & a^{173} \\
a^{240} & a^{41} & a^{81} & a^{10} & a^{220} & a^{90} & a^{78} & a^{147} & a^{67} & a^{247} & a^{120} \\
a^{28} & a^{241} & a^{141} & a^{158} & a^{201} & a^{145} & a^{212} & a^{20} & a^{20} & a^{130} & a^{42} \\
a^{226} & a^{39} & a^{104} & a^{162} & a^{136} & a^{209} & a^{40} & a^{165} & a^{91} & a^{128} & a^{110} \\
a^{243} & a^{3} & a^{208} & a^{15} & a^{174} & a^{92} & a^{212} & a^{184} & a^{142} & a^{64} & a^{183} \\
a^{25} & a^{69} & a^{228} & a^{161} & a^{230} & a^{123} & a^{214} & a^{159} & a^{90} & a^{161} & a^{217} \\
a^{161} & a^{226} & a^{146} & a^{115} & a^{187} & a^{55} & a^{64} & a^{46} & a^{48} & a^{32} & a^{198} \\
\end{smallmatrix}
\right).
\]
	By using SageMath, we find out that every set of \( 9 \) columns of \( \textbf H \) is linearly independent, which means that the minimum distance of \( \mathcal C \) is at least \( 10 \), while there exist sets of \( 10 \) linearly dependent columns. Hence, \( \mathcal C \) is a \( [20, 9, 10] \) code.

	By Proposition \ref{prop:dual}, \( \mathcal C^\perp \) is the skew \( \lambda^{-1} \)-constacyclic code \( \widehat{\mathcal R} h^\Theta \), where \( \widehat{\mathcal R} = R/R(x^{20} - \lambda^{-1}) \) and \( h g = x^{20} - \lambda \). The generator of \( \mathcal C^\perp \) is \( h^\Theta = x^{9} + a^{49}x^8 + a^{15} x^7 + a^{122} x^6 + a^{54} x^5 + a^{27} x^4 + a^{110} x^3 + a^{61} x^2 + a^{233} x + a^{147} \). In order to get a \( \lambda \)-constacyclic code whose dual is isomorphic to the one of \( \mathcal C^\perp \), we use an isomorphism of the form \( \varphi_\beta \) for some \( \beta \in \units F \), as described in Lemma \ref{lemma:phibtocalR}, such that \( \norm{F/K}(\beta)^s = \lambda^2 \). As \( \norm{F/K} \) is a map from \( F \) to \( K \) and the map \( \gamma \mapsto \gamma^s = \gamma^5 \) is the Frobenius endomorphism in \( K \) since \( \gamma^4 = \gamma \) for all \( \gamma \in K \), any element \( u \in \units F \) such that \( \norm{F/K}(u) = \lambda \), which must exist as the norm over finite field extensions is surjective, is also such that \( \norm{F/K}(u)^s = \lambda^2 \). It turns out that \( \norm{F/K}(a) = \lambda \) and therefore we can use \( \varphi_a : \widehat{\mathcal R} \to \mathcal R \) for our purpose. The monic skew polynomial that generates the same code as \( \varphi_a(h^\Theta) \) is \( h' = x^{9} + a^{50}x^8 + a^{80}x^7 + a^{203} x^6 + a^{139} x^5 + a^{113} x^4 + a^5 x^3 + a^{227} x^2 + a^{148} x + a^{63} \). Its greatest common right divisor with \( g \) is \( 1 \), so \( (\mathcal Rg, \mathcal Rh') \) is directly an LCP of codes with security parameter \( 10 \).

	Through applying the automorphism group given in Proposition \ref{prop:automorphismgroup} to \( h' \), we get \( 340 \) possible supplements for \( \mathcal Rg \), \( 200 \) of them resulting in an LCP of codes.
\end{example}

Another option is to generate \( \mathcal D^\perp \) (or \( \mathcal C \) if we had \( \mathcal D \)) as a skew \( \lambda^{-1} \)-constacyclic (respectively, \( \lambda \)-constacyclic) code with a high enough lower bound for its distance. A way to do so, if the dimension needed for the code is divisible by \( s \), is through Definition \ref{defnBCHlike}, which gives a skew BCH code. If the code must have dimension \( k = s k' \), then the skew BCH code will have \( \mu - k' + 1 \), where \( \mu = |\sigma| \), as its designed distance. If \( s = 1 \), this will actually yield MDS codes, so the security parameter is guaranteed not to be limited by this added code. In contrast, the resulting distance might be unsatisfactory for \( s > 1 \), depending on the desired application and the first code.

\begin{example}\label{ex:9ht}
	Consider as \( \mathcal C = \mathcal R g \) the ninth code in \cite[Table 1]{GLNN18}, which is the skew cyclic \( F \)-linear code of length \( n = 16 \) and dimension \( k = 8 \) for \( L = \FF_{2^{16}} = \FF_2(a) \) where \( a^{16} + a^{11} + a^2 + a + 1 = 0 \), \( \theta : L \to L \) is the Frobenius endomorphism (hence \( K = L^\theta = \FF_2, n = 16 \)), \( F = L^{\theta^8} = \FF_{2^8} \) (so \( \mu = 8, s = 2, \sigma = \theta_{|F} \)), \( \alpha = a^5 \)
	and
	\[
		g = \lclm{x - \theta^{i + 8j}(\alpha^{-1}\theta(\alpha)) ~|~ i \in \{ 0, 1, 3, 4 \}, 0 \le j \le 1}.
	\]

	While Proposition \ref{propBCHlike} only guarantees a minimum distance for \( \mathcal C = \mathcal R g \) of at least \( 3 \), it is at least \( 4 \) by the Hartmann--Tzeng bound described in \cite{GLNN18}. We may get a pair whose security parameter matches \( \distance(\mathcal C) \) whatever its actual value happens to be by repeating the steps in Example \ref{ex:1roos}. Here, we shall instead construct the supplement \( \mathcal D \) from a skew BCH cyclic code.

	We will use \( b = a^8 + a^7 + a^3 + a^2 + 1 \) as the element such that \( F = \FF_2(b) \), so \( g = x^8 + b^{110}x^7 + b^{125}x^6 + b^{219}x^5 + b^{101}x^4 + b^{191}x^3 + b^{45}x^2 + b^{67}x + b^{85} \).
	The skew polynomial
	\[
		h = \lclm{x - \theta^{i + 8j}(\alpha^{-1}\theta(\alpha)) ~|~ 0 \le i \le 3, 0 \le j \le 1}
	\]
	is such that the minimum weight of the \( [16, 8] \) code \( \mathcal R h \) is at least \( 5 \) by Proposition \ref{propBCHlike}. Now we define \( \mathcal D = (\mathcal R h)^\perp \), which by Proposition \ref{prop:dual} is \( \mathcal R \hat h \) where \( \hat h = x^8 + b^{235}x^7 + b^{222}x^6 + b^{178}x^5 + b^{17}x^4 + b^{111}x^3 + b^{71}x^2 + b^{194}x + b^{198} \) is the result of left multiplying \( x^8\Theta(h') \) by the inverse of its leading coefficient, and \( h' \) is such that \( h' h = x^{16} - 1 \).

	The codes \( \mathcal R g \) and \( \mathcal R \hat h \) are not an LCP of codes, as \( \gcrd{g, \hat h} \ne 1 \). The group from Proposition \ref{prop:automorphismgroup}, which has \( 2040 \) elements, results in \( 672 \) possible supplements for \( \mathcal R g \). For example, the only one among those which is given by a power of \( \phi_x \) is the automorphism \( \phi_x^6 \). This yields \( \phi_x^6(\hat h) = x^8 + b^{250}x^7 + b^{183}x^6 + b^{172}x^5 + b^{68}x^4 + b^{219}x^3 + b^{209}x^2 + b^{176}x + b^{177} \), which satisfies \( \gcrd{g, \phi_x^6(\hat h)} = 1 \). Consequently, \( (\mathcal R g, \mathcal R \phi_x^6(\hat h)) \) is an LCP of codes with security parameter at least \( 4 \).

	The size of this example is small enough to check the exact distance of the codes through their parity check matrices, as it was done in Example \ref{ex:20,9,10}. For both \( \mathcal R g \) and \( (\mathcal R \hat h)^\perp \), this distance is \( 8 \), so the actual security parameter of the pairs of codes in this example is \( 8 \). However, note that using a BCH code only guarantees a distance of \( 5 \) in this case.
\end{example}

If the dimension of the chosen code is half its length and its distance and dual distance are equal (which can be known beforehand if the code is MDS, as the dual of an MDS code is an MDS code, see \cite[Theorem 2.4.3]{HP03}), we could complete the pair by repeating the same code. In this case, we will need the automorphism group given by Proposition \ref{prop:automorphismgroup}, as it is immediate that the pair will not be an LCP of codes.

\begin{example}\label{ex:12,6,6}
	Take the same \( \sigma \), \( F \), \( a \), \( \lambda \) and \( R \) as in Example \ref{ex:20,9,10} and define \( \mathcal R = R/R(x^{12} - \lambda) \). In this case, left ideals of \( \mathcal R \) are skew \( \lambda \)-constacyclic codes of length \( n = 12 \). As the skew polynomial \( g = x^{6} + a^{16} x^{5} + a^{131} x^{4} + a^{159} x^{3} + a^{46} x^{2} + a^{61} x + a^{218} \) right divides \( x^{12} - \lambda \), \( \mathcal C = \mathcal R g \) is a \( [12, 6] \) \( F \)-linear skew \( \lambda \)-constacyclic code.
	
	Note that, as \( s = 3 \) and \( \gamma^3 = 1 \) for every \( \gamma \in \units K \), the map \( \left(\norm{F/K}\right)^s : F \to K \) sends every nonzero element into \( 1 \). In particular, there is no \( \beta \in \units F \) such that an isomorphism \( \varphi_\beta : R/R(x^{12} - \lambda^{-1}) \to \mathcal R \), as given in Lemma \ref{lemma:phibtocalR}, exists.

	A parity-check matrix for \( \mathcal C \) is
\[
\setcounter{MaxMatrixCols}{12}
\textbf H =
\begin{pmatrix}
1 & 0 & 0 & 0 & 0 & 0 & a^{218} & a^{27} & a^{139} & a^{65} & a^{184} & a^{237} \\
0 & 1 & 0 & 0 & 0 & 0 & a^{61} & a^{92} & a^{92} & a^{168} & a^{240} & a^{42} \\
0 & 0 & 1 & 0 & 0 & 0 & a^{46} & a^{32} & a^{184} & a^{145} & a^{155} & a^{134} \\
0 & 0 & 0 & 1 & 0 & 0 & a^{159} & a^{35} & a^{78} & a^{3} & a^{133} & a^{127} \\
0 & 0 & 0 & 0 & 1 & 0 & a^{131} & a & a^{227} & a^{152} & a^{182} & a^{162} \\
0 & 0 & 0 & 0 & 0 & 1 & a^{16} & a^{44} & a^{153} & a^{119} & a^{196} & a^{158} \\
\end{pmatrix},
\]
	which can be checked to have no sets of \( 5 \) linearly dependent columns while it does have sets of \( 6 \) linearly dependent columns, so \( \distance(\mathcal C) = 6 \). This distance cannot be matched through the designed distance of a skew BCH code as given by Definition \ref{defnBCHlike}. However, we can pick any skew \( \lambda \)-constacyclic code \( \mathcal D \) such that \( \dim_F(\mathcal C) + \dim_F(\mathcal D) = n \) and check the distance of its dual code using the same procedure as we did for \( \mathcal C \). In this case \( n = 2 \dim_F(\mathcal C) \), so \( \dim_F(\mathcal D) \) must match \( \dim_F(\mathcal C) \). As a result, we may try \( \mathcal D = \mathcal C \). The matrix
	\[
\textbf G = 
\begin{pmatrix}
a^{218} & a^{61} & a^{46} & a^{159} & a^{131} & a^{16} & 1 & 0 & 0 & 0 & 0 & 0 \\
0 & a^{107} & a^{244} & a^{184} & a^{126} & a^{14} & a^{64} & 1 & 0 & 0 & 0 & 0 \\
0 & 0 & a^{173} & a^{211} & a^{226} & a^{249} & a^{56} & a & 1 & 0 & 0 & 0 \\
0 & 0 & 0 & a^{182} & a^{79} & a^{139} & a^{231} & a^{224} & a^{4} & 1 & 0 & 0 \\
0 & 0 & 0 & 0 & a^{218} & a^{61} & a^{46} & a^{159} & a^{131} & a^{16} & 1 & 0 \\
0 & 0 & 0 & 0 & 0 & a^{107} & a^{244} & a^{184} & a^{126} & a^{14} & a^{64} & 1 \\
\end{pmatrix}
\]
	is a generator matrix for \( \mathcal C \) (note that the \(i\)-th row has the coefficients of \( x^{i-1} g \)) and therefore a parity-check matrix for \( \mathcal C^\perp \). We can check that no set of \( 5 \) linearly dependent columns exists, while some sets of \( 6 \) columns are dependent, hence \( \distance(\mathcal C^\perp) = 6 \).
	
	It is immediate that, while the pair \( (\mathcal C, \mathcal C) \) has \( 6 \) as its security parameter, it is not a linear complementary pair of codes. This is sorted out by applying \( \phi_x^k(g) \), for any \( k \in \{ 1, 2, 3 \} \), to one of the elements in the pair, as then the greatest common right divisor of the generator skew polynomials becomes \( 1 \) and Theorem \ref{thm:LCP} applies. We can take advantage of the fact that, as \( \phi_x^k(g) = x^k g x^{-k} \), the \( (k + 1) \)-th row of the generator matrix \( \textbf G \) contains the coefficients of \( \phi_x^k(g) \) moved \( k \) columns to the right. For example, from the second (or last) row we conclude that \( \phi_x(g) = x^6 + a^{64}x^5 + a^{14}x^4 + a^{126}x^3 + a^{184} x^2 + a^{244}x + a^{107} \).

	We may also consider the automorphisms of the form \( \varphi_\beta \) for \( \beta \in \units F \). As previously shown, \( \gamma^s = 1 \ \forall \gamma \in K \), so \( \varphi_\beta \) is always an isomorphism \( \mathcal R \to \mathcal R \). This results, if compositions with powers of \( \phi_x \) are also used, in \( 1019 \) isomorphisms to consider, excluding the identity map. Among those, no pair of them results in the same generator and only \( 8 \) fail to yield an LCP of codes, so there are \( 1011 \) possible LCPs of codes of the form \( (\mathcal C, \isom(\mathcal C)) \) where \( \isom \) is an isomorphism from the group described in Proposition \ref{prop:automorphismgroup}, and there are \( 1020 \cdot 1011 = 1\,031\,220 \) LCPs of codes of the form \( (\isom_1(\mathcal C),\isom_2(\mathcal C)) \), all of them with security parameter \( 6 \), since we can also apply any of the \( 1020 \) available isomorphisms to the whole pair.
	
	Note that, while \( \distance(\mathcal C) = \distance(\mathcal C^\perp) \), this is not the case for every code of length \( 6 \) in our context. For example, \( p = x^6 + a^{24}x^5 + a^{183}x^4 + a^{164}x^3 + a^{82}x^2 + a^{70}x + a^{89} \) is a right divisor of \( x^{12} - \lambda \) which generates a \( [12, 6, 6] \) skew \( \lambda \)-constacyclic code whose dual is a \( [12, 6, 4] \) code. In fact, \( (\mathcal R p, \mathcal R \phi_x(p)) \) is an LCP of codes. This shows that, in general, for an LCP of codes \( (\mathcal C, \mathcal D) \), it is not guaranteed that \( \distance(\mathcal C) = \distance(\mathcal D^\perp) \), as it is the case for non-skew cyclic codes, see \cite[Theorem 2.4]{Guneri18}.
\end{example}

\begin{table}[h]
\caption{Parameters for the examples above, where \( k \) stands for the dimension over \( F \) of the first element in each LCP of codes and \( \text{sp} \) is the security parameter for the resulting LCPs.}\label{table:examples}
\begin{tabular}{c|r|r|r|c|r}
\text{Ex.} & $n$ & $k$ & sp & $|K|^\mu$ & $s$ \\
\midrule
\ref{ex:1roos}   & 12 &  6 &  6 & $2^6$ &  2 \\
\ref{ex:9bgu}    & 44 & 20 & 17 & $3^2$ & 22 \\
\ref{ex:8roos}   & 10 &  6 &  5 & $5^5$ &  2 \\
\ref{ex:20,9,10} & 20 &  9 & 10 & $4^4$ &  5 \\
\ref{ex:9ht}     & 16 &  8 &  8 & $2^8$ &  2 \\
\ref{ex:12,6,6}  & 12 &  6 &  6 & $4^4$ &  3
\end{tabular}
\end{table}

Table \ref{table:examples} summarizes the parameters of the discussed examples, including the length \( n \), the cardinality of the field \( F \) (shown as \( |K|^\mu \)) and the security parameter. Note that no effort was made to maximize the security parameter with respect to \( n \), \( k \) and \( |F| \) as the goal was to illustrate some available options to construct LCPs of codes by completing a given code into a pair. This table is not meant to show upper bounds for the possible security parameters.

\end{document}